\documentclass[a4paper,USenglish]{lipics-v2021}

\usepackage[dvipsnames,table]{xcolor}
\usepackage{tikz}

\usepackage{tlatex}
\usepackage{hyperref}
\usepackage{caption}
\usepackage{subcaption}
\usepackage{amsmath}
\usepackage{multicol}
\usepackage{mathtools}
\usepackage{multirow,bigdelim}
\usepackage{makecell}
\usepackage{amsthm}
\usepackage{skak}
\usepackage{algorithm}
\usepackage[noend]{algpseudocode}
\usepackage[textsize=tiny]{todonotes}

\newif\ifextended
\extendedtrue

\title{Design and Analysis of a Logless Dynamic Reconfiguration Protocol}

\titlerunning{Logless Dynamic Reconfiguration} %TODO optional, please use if title is longer than one line

\author{William Schultz}{Northeastern University, USA}{schultz.w@northeastern.edu}{}{}
\author{Siyuan Zhou}{MongoDB, USA}{siyuan.zhou@mongodb.com}{}{}
\author{Ian Dardik}{Northeastern University, USA}{dardik.i@northeastern.edu}{}{}
\author{Stavros Tripakis}{Northeastern University, USA}{s.tripakis@northeastern.edu}{}{}

\authorrunning{W. Schultz, S.Zhou, et al.} %TODO mandatory. First: Use abbreviated first/middle names. Second (only in severe cases): Use first author plus 'et al.'

\Copyright{William Schultz, Siyuan Zhou, Ian Dardik, Stavros Tripakis} %TODO mandatory, please use full first names. LIPIcs license is "CC-BY";  http://creativecommons.org/licenses/by/3.0/


 %TODO mandatory; please add comma-separated list of keywords

\supplement{TLA+ specifications \cite{supp-materials}: \url{https://doi.org/10.5281/zenodo.5715510}}%optional, e.g. related research data, source code, ... hosted on a repository like zenodo, figshare, GitHub, ...
\acknowledgements{We would like to thank Tess Avitabile for her critical insights during the development of the reconfiguration protocol and discovery of subtle bugs in early design proposals. We would like to thank Judah Schvimer, A. Jesse Jiryu Davis, Pavi Vetriselvan, and Ali Mir for offering helpful insights during the protocol design and implementation process. We would also like to thank Shuai Mu for providing helpful comments on initial drafts of this paper.}%optional

\nolinenumbers %uncomment to disable line numbering

\hideLIPIcs  %uncomment to remove references to LIPIcs series (logo, DOI, ...), e.g. when preparing a pre-final version to be uploaded to arXiv or another public repository

\EventEditors{Quentin Bramas, Vincent Gramoli, and Alessia Milani}
\EventNoEds{3}
\EventLongTitle{25th International Conference on Principles of Distributed Systems (OPODIS 2021)}
\EventShortTitle{OPODIS 2021}
\EventAcronym{OPODIS}
\EventYear{2021}
\EventDate{December 13--15, 2021}
\EventLocation{Strasbourg, France}
\EventLogo{}
\SeriesVolume{217}
\ArticleNo{18}

\begin{document}

%%%
%%% PODC 2021 submission info.
%%%

% %%
% %% The "title" command has an optional parameter,
% %% allowing the author to define a "short title" to be used in page headers.
% \title{Design and Verification of a Logless Dynamic Reconfiguration Protocol in MongoDB Replication}

% %%
% %% The "author" command and its associated commands are used to define
% %% the authors and their affiliations.
% %% Of note is the shared affiliation of the first two authors, and the
% %% "authornote" and "authornotemark" commands
% %% used to denote shared contribution to the research.
% \author{William Schultz}
% %\authornote{Both authors contributed equally to this research.}
% \email{schultz.w@northeastern.edu}
% \affiliation{%
%   \institution{Northeastern University}
%   \streetaddress{}
%   \city{Boston}
%   \state{MA}
%   \country{United States} 
%   \postcode{}
% }

% \author{Siyuan Zhou}
% %\authornote{Both authors contributed equally to this research.}
% \email{siyuan.zhou@mongodb.com}
% %\authornotemark[1]
% \affiliation{%
%   \institution{MongoDB, Inc.}
%   \streetaddress{}
%   \state{NY}
%   \country{United States} 
%   \postcode{}
% }

% \author{Stavros Tripakis}
% %\authornote{Both authors contributed equally to this research.}
% \email{s.tripakis@northeastern}
% %\authornotemark[1]
% \affiliation{%
%   \institution{Northeastern University}
%   \streetaddress{}
%   \city{Boston}
%   \state{MA}
%   \country{United States} 
%   \postcode{}
% }

\maketitle

\begin{abstract}

    Distributed replication systems based on the replicated state machine model have become ubiquitous as the foundation of modern database systems. To ensure availability in the presence of faults, these systems must be able to dynamically replace failed nodes with healthy ones via \emph{dynamic reconfiguration}. 
    MongoDB is a document oriented database with a distributed replication mechanism derived from the Raft protocol. 
    In this paper, we present \textit{MongoRaftReconfig}, a novel dynamic reconfiguration protocol for the MongoDB replication system.
    % We present a novel dynamic reconfiguration protocol for the MongoDB replication system that extends and generalizes the single server reconfiguration protocol of the Raft consensus algorithm. 
    \textit{MongoRaftReconfig} utilizes a logless approach to managing configuration state and decouples the processing of configuration changes from the main database operation log. 
    The protocol's design was influenced by engineering constraints faced when attempting to redesign an unsafe, legacy reconfiguration mechanism that existed previously in MongoDB.
    We provide a safety proof of \textit{MongoRaftReconfig}, along with a formal specification in TLA+. To our knowledge, this is the first published safety proof and formal specification of a reconfiguration protocol for a Raft-based system. We also present results from model checking the safety properties of \textit{MongoRaftReconfig} on finite protocol instances. 
    %
    % The design of \textit{MongoRaftReconfig} was heavily influenced by engineering constraints faced when attempting to repair an unsafe, legacy reconfiguration mechanism that existed previously in the MongoDB system.
    % The design of \textit{MongoRaftReconfig} was influenced by the engineering constraints faced when attempting to repair an unsafe, legacy reconfiguration mechanism that existed previously in the MongoDB system.
    % We discuss how we used formal specification and model checking tools to aid in this repair process, which enabled an iterative, counterexample guided design process for \textit{MongoRaftReconfig}.
    % design and repair via a counterexample guided design process.
    %
    Finally, we discuss the conceptual novelties of \textit{MongoRaftReconfig}, how it can be understood as an optimized and generalized version of the single server reconfiguration algorithm of Raft, and present an experimental evaluation of how its optimizations can provide performance benefits for reconfigurations.
    
    % , and how it allows reconfigurations to proceed in cases when the main log is prevented from processing new operations.

% We also provide an experimental evaluation of the protocol benefits, showing how reconfigurations are able to quickly restore a system to healthy operation in scenarios where node failures have stalled the main operation log.

% allows reconfigurations to proceed in cases when the main log is prevented from processing new operations. Additionally, this decoupling allows for configuration state to be managed by a \textit{logless} replicated state machine, by optimizing away the explicit log and storing only the latest version of the configuration, avoiding the complexities of a log-based protocol. 

\end{abstract}

\newcommand{\divider}{\noindent\rule{12cm}{0.4pt}\\}

\section{Introduction}
\label{sec:intro}
Distributed replication systems based on the replicated state machine model \cite{Schneider1990} have become ubiquitous as the foundation of modern, fault-tolerant data storage systems. In order for these systems to ensure availability in the presence of faults, they must be able to dynamically replace failed nodes with healthy ones, a process known as \emph{dynamic reconfiguration}. The protocols for building distributed replication systems have been well studied and implemented in a variety of systems \cite{paxosmadelive,Corbett2012,Huang2020,cockroach}. Paxos \cite{Lamport1998} and, more recently, Raft \cite{raftpaper}, have served as the logical basis for building provably correct distributed replication systems.  Dynamic reconfiguration, however, is an additionally challenging and subtle problem \cite{aguilera2010} that has not been explored as extensively as the foundational consensus protocols underlying these systems. Variants of Paxos have examined the problem of dynamic reconfiguration but these reconfiguration techniques may require changes to a running system that impact availability \cite{malkhi2008stoppable} or require the use of an external configuration master \cite{vertical-paxos}. The Raft consensus protocol, originally published in 2014, provided a dynamic reconfiguration algorithm in its initial publication, but did not include a precise discussion of its correctness or include a formal specification or proof. A critical safety bug \cite{ongaro2015-membership-bug} in one of its reconfiguration protocols was found after initial publication, demonstrating that the design and verification of reconfiguration protocols for these systems is a challenging task.

% The discovery of bugs like these demonstrate that the design and verification of a safe dynamic reconfiguration protocol is a non-trivial task, and the correctness of published protocols may not be particularly well understood by system designers and engineers, which is often important if optimizations or modifications need to be made to an underlying protocol. As a rule of thumb, we believe that if distributed systems researchers can make errors designing these protocols, it is even harder for system engineers to understand the subtleties of these protocols when implementing them. Thus, it is important for their behaviors and characteristics to be well studied, formalized, and understood in different contexts and systems. 

MongoDB \cite{mongodb-repo} is a general purpose, document oriented database which implements a distributed replication system \cite{schultz2019tunable} for providing high availability and fault tolerance. MongoDB's replication system uses a novel consensus protocol that derives from Raft \cite{zhou2021fault}. Since its inception, the MongoDB replication system has provided a custom, legacy protocol for dynamic reconfiguration of replica members that was not based on a published algorithm.
% This legacy protocol was sufficient to provide basic reconfiguration functionality to clients, but it was known to be unsafe in certain cases. 
%
This legacy protocol managed configurations in a \textit{logless} fashion i.e. each server only stored its latest configuration. In addition, it decoupled reconfiguration processing from the main database operation log.
% managed configurations without a log-based mechanism, and so reconfiguration processing was completely decoupled from the main database operation log. 
These features made for a simple and appealing protocol design, and it was sufficient to provide basic reconfiguration functionality to clients. The legacy protocol, however, was known to be unsafe in certain cases.
%  it had outstanding safety issues.
% The legacy protocol, however, had a simple design, in part due to the fact that it operated independently from the main database operation log, allowing it to operate in a \textit{logless} fashion, by gossiping configurations between members of the replica set. 
In recent versions of MongoDB, reconfiguration has become a more common operation, necessitating the need for a redesigned, safe reconfiguration protocol with rigorous safety guarantees. 
From a system engineering perspective, a primary goal was to keep design and implementation complexity low. Thus, it was desirable that the new reconfiguration protocol minimize changes to the legacy protocol to the extent possible. 
%
% During the early phases of design for the new protocol, we conjectured it may be possible to retain the underlying architecture of the legacy reconfiguration protocol and make it safe with relatively small, localized revisions.
% 
%That is, we felt it would be possible to \textit{repair} the legacy protocol to produce a safe protocol with minimal modifications. 
%In this paper, we present \textit{MongoRaftReconfig}
%
% Cut or keep the below sentence?
%
% To achieve this, we developed a formal specification of the legacy protocol in TLA+ \cite{Merz2008}, so that we could accurately classify and debug its safety issues precisely. This allowed us to efficiently understand exactly how the protocol was broken, giving us intuition on how to design a new protocol with minimal modifications. 
% It also afforded us the prototyping speed and confidence necessary to design a new, safe protocol in a relatively short period of time i.e. no more than a few weeks.
% The legacy protocol, however, had a simple design, in part due to the fact that it operated independently from the main database operation log, allowing it to operate in a \textit{logless} fashion, by gossiping configurations between members of the replica set. 
% So, we focused our design efforts on revising the legacy protocol into a new, safe protocol while minimzing changes. 
In this paper, we present \textit{MongoRaftReconfig}, a novel dynamic reconfiguration protocol that achieves the above design goals. 

\textit{MongoRaftReconfig} provides safe, dynamic reconfiguration, 
utilizes a logless approach to managing configuration state, and
decouples reconfiguration processing from the main database operation log. Thus, it bears a high degree of architectural and conceptual similarity to the legacy MongoDB protocol, satisfying our original design goal of minimizing changes to the legacy protocol. 
We provide rigorous safety guarantees of \textit{MongoRaftReconfig}, including a proof of the protocol's main safety properties along with a formal specification in TLA+ \cite{Merz2008}, a specification language for describing distributed and concurrent systems. To our knowledge, this is the first published safety proof and formal specification of a reconfiguration protocol for a Raft-based system. We also verified the safety properties of finite instances of \textit{MongoRaftReconfig} using the TLC model checker \cite{yu1999model}, which provides additional confidence in its correctness. Finally, we discuss the conceptual novelties of \textit{MongoRaftReconfig}, related to its logless design and decoupling of reconfiguration processing. In particular, we discuss how it can be understood as an optimized and generalized variant of the single server Raft reconfiguration protocol. We also include a preliminary experimental evaluation of how these optimizations can provide performance benefits over standard Raft, by allowing reconfigurations to bypass the main operation log.

To summarize, in this paper we make the following contributions:
\begin{itemize}
    \item We present \textit{MongoRaftReconfig}, a novel, logless dynamic reconfiguration protocol for the MongoDB replication system.

    % \textcolor{blue}{To our knowledge, this is the first complete published formal specification of a Raft based reconfiguration protocol.}
    
    \item We present a proof of \textit{MongoRaftReconfig}'s key safety properties. To our knowledge, this is the first published safety proof of a reconfiguration protocol for a Raft-based system.
    
    \item We present a formal specification of \textit{MongoRaftReconfig} in TLA+. To our knowledge, this is the first published formal specification of a reconfiguration protocol for a Raft-based system.

    \item We present results of model checking the safety properties of \textit{MongoRaftReconfig} on finite protocol instances using the TLC model checker. 

    % \textcolor{blue}{We present a proof of \textit{MongoRaftReconfig}'s key safety properties. To our knowledge, this is the first published safety proof of a Raft based reconfiguration algorithm.}
    
    % \item We discuss our process of counterexample guided protocol design that was aided by the formal TLA+ specification and the TLC model checker.

    \item We discuss the conceptual novelties of \textit{MongoRaftReconfig}, and how it can be understood as an optimized and generalized variant of the single server Raft reconfiguration protocol.
    
    % the concepts and behaviors of the protocol can be understood in the context of reconfiguration in standard Raft.
    
    % , and how these aspects provide potential reconfiguration performance benefits.
    
    % Specifically, we show our how protocol optimizes Raft reconfiguration by avoiding unnecessary commitment of writes during reconfigurations, and how it simplifies the log structure for managing configuration state.
    
    % TODO: Cut or keep?
    \item We provide a preliminary experimental evaluation of \textit{MongoRaftReconfig}'s performance benefits, demonstrating how it improves upon reconfiguration in standard Raft.

\end{itemize}

\section{Background}

\subsection{System Model}

Throughout this paper, we consider a set of \textit{server} processes $Server=\{s_1,s_2,...,s_n\}$ that communicate by sending messages. We assume an asynchronous network model in which messages can be arbitrarily dropped or delayed. We assume servers can fail by stopping but do not act maliciously i.e. we assume a ``fail-stop" model with no Byzantine failures. We define both a \textit{member set} and a \textit{quorum} as elements of $2^{Server}$. Member sets and quorums have the same type but refer to different conceptual entities.
% which contains all quorums in $\mathcal{P}(m)$ with at least a majority of elements in $m$.
For any member set $m$, and any two non-empty member sets $m_i,m_j$, we define the following:
\begin{align}
   &Quorums(m) \triangleq \{s \in 2^{m} : |s| \cdot 2 > |m|\} \\
   \label{def:quorums-overlap}
   &QuorumsOverlap(m_i,m_j) \triangleq
     \forall q_i \in Quorums(m_i), q_j \in Quorums(m_j) : q_i \cap q_j \neq \emptyset
\end{align}
where $|S|$ denotes the cardinality of a set $S$. We refer to Definition \ref{def:quorums-overlap} as the \textit{quorum overlap} condition.

\subsection{Raft}

Raft \cite{OngaroDissertation2014} is a consensus protocol for implementing a replicated log in a system of distributed servers. It has been implemented in a variety of systems across the industry \cite{raft-website}. Throughout this paper, we refer to the original Raft protocol as described and specified in \cite{OngaroDissertation2014} as \textit{standard Raft}. 

The core Raft protocol implements a replicated state machine using a static set of servers. In the protocol, time is divided into \textit{terms} of arbitrary length, where terms are numbered with consecutive integers. Each term has at most one leader, which is selected via an \textit{election} that occurs at the beginning of a term. To dynamically change the set of servers operating the protocol, Raft includes two, alternate algorithms: \textit{single server membership change} and \textit{joint consensus}. 
In this paper we are only concerned with \textit{single server membership change}. The single server change approach aims to simplify reconfiguration by allowing only reconfigurations that add or remove a single server. Reconfiguration is accomplished by writing a special reconfiguration entry into the main Raft operation log that alters the local configuration of a server.  In this paper, when referring to reconfiguration in standard Raft, we assume it to mean the single server change protocol. 

\subsection{Replication in MongoDB}

MongoDB is a general purpose, document oriented database that stores data in JSON-like objects. A MongoDB database consists of a set of collections, where a collection is a set of unique documents. To provide high availability, MongoDB provides the ability to run a database as a \textit{replica set}, which is a set of MongoDB servers that act as a consensus group, where each server maintains a logical copy of the database state. 
% \subsubsection{Static Replication Protocol} 
% \label{sec:static-repl-desc}

MongoDB replica sets utilize a replication protocol that is derived from Raft, with some extensions. We refer to MongoDB's abstract replication protocol, without dynamic reconfiguration, as \textit{MongoStaticRaft}. This protocol can be viewed as a modified version of standard Raft that satisfies the same underlying correctness properties. A more in depth description of \textit{MongoStaticRaft} is given in \cite{zhou2021fault,schultz2019tunable}, but we provide a high level overview here, since the \textit{MongoRaftReconfig} reconfiguration protocol is built on top of \textit{MongoStaticRaft}. In a replica set running \textit{MongoStaticRaft} there exists a single \textit{primary} server and a set of \textit{secondary} servers. As in standard Raft, there is a single primary elected per term. The primary server accepts client writes and inserts them into an ordered operation log known as the \textit{oplog}. The oplog is a logical log where each entry contains information about how to apply a single database operation. Each entry is assigned a monotonically increasing timestamp, and these timestamps are unique and totally ordered within a server log. These log entries are then replicated to secondaries which apply them in order leading to a consistent database state on all servers. When the primary learns that enough servers have replicated a log entry in its term, the primary will mark it as \textit{committed}, guaranteeing that the entry is permanently durable in the replica set.
% TODO: May only need these sentences if we end up including the experimental evaluation.
Clients of the replica set can issue writes with a specified \textit{write concern} level, which indicates the durability guarantee that must be satisfied before the write can be acknowledged to the client. Providing a write concern level of \textit{majority} ensures that a write will not be acknowledged until it has been marked as committed in the replica set. A key, high level safety requirement of the replication protocol is that if a write is acknowledged as committed to a client, it should be durable in the replica set.

\section{MongoRaftReconfig: A Logless Dynamic Reconfiguration Protocol}
\label{sec:mongo-raft-reconfig-protocol}

In this section we present the \textit{MongoRaftReconfig} dynamic reconfiguration protocol. First, we provide an overview and some intuition on the protocol design in Section \ref{sec:protocol-intuition}. Section \ref{sec:protocol-behavior} provides a high level, informal description of the protocol along with a condensed pseudocode description in Algorithm \ref{alg:mrr-pseudocode}. Sections \ref{sec:reconfig-safety-restrictions} and \ref{sec:config-and-elections} provide additional detail on the mechanisms required for the protocol to operate safely, and the TLA+ formal specification of \textit{MongoRaftReconfig} is discussed briefly in Section \ref{sec:formal-spec}.

\ifextended
The complete description of \textit{MongoRaftReconfig} is left to Algorithm \ref{alg:mrr-pseudocode-full} in Appendix \ref{appendix:safety-proof}.
\else
The complete description of \textit{MongoRaftReconfig} is left to the full version of the paper \cite{schultz2021design}.
\fi
The pseudocode presented in Algorithm \ref{alg:mrr-pseudocode} describes the reconfiguration specific behaviors of \textit{MongoRaftReconfig}, which are the novel aspects of the protocol and the contributions of this paper.

% and an overview of the safety proof is presented in Section \ref{sec:correctness}. Section \ref{sec:raft-comparison} discusses the conceptual novelties of \textit{MongoRaftReconfig} and how it compares to standard Raft. Section \ref{sec:experiments} presents a preliminary experimental evaluation of its performance benefits.

% In Section \ref{sec:ceg-design}, we discuss aspects of our counterexample guided approach to protocol design, which also serves to provide further intuition on the protocol and how it operates safely.

% We then provide additional intuition on the protocol and its design, which stemmed from the process of repairing MongoDB's legacy reconfiguration protocol via counterexample analysis aided by the model checker.

\subsection{Overview and Intuition}
\label{sec:protocol-intuition}

Dynamic reconfiguration allows the set of servers operating as part of a replica set to be modified while maintaining the core safety guarantees of the replication protocol. Many consensus based replication protocols \cite{Shraer2019,malkhi2008stoppable,raftpaper} utilize the main operation log (the \textit{oplog}, in MongoDB) to manage configuration changes by writing special reconfiguration log entries. The \textit{MongoRaftReconfig} protocol instead decouples configuration updates from the main operation log by managing the configuration state of a replica set in a separate, logless replicated state machine, which we refer to as the \textit{config state machine}. The config state machine is maintained alongside the oplog, and manages the configuration state used by the overall protocol.

In order to ensure safe reconfiguration, \textit{MongoRaftReconfig} imposes specific restrictions on how reconfiguration operations are allowed to update the configuration state of the replica set. First, it imposes a \textit{quorum overlap} condition on any reconfiguration from $C$ to $C'$, which is an approach adopted from the Raft single server reconfiguration algorithm.  This ensures that all quorums of two adjacent configurations overlap with each other, and so can safely operate concurrently. In order to allow the system to pass through many configurations over time, though, \textit{MongoRaftReconfig} imposes additional restrictions which address two essential aspects required for safe dynamic reconfiguration: (1) \textit{deactivation} of old configurations and (2) \textit{state transfer} from old configurations to new configurations. Essentially, it must ensure that old configurations, which may not overlap with newer configurations, are appropriately prevented from executing disruptive operations (e.g. electing a primary or committing a write), and it must also ensure that relevant protocol state from old configurations is properly transferred to newer configurations before they become active. The details of these restrictions and their safety implications are discussed further in Section \ref{sec:reconfig-safety-restrictions}.

In the remainder of this section we give an overview of the behaviors of \textit{MongoRaftReconfig}, along with a pseudocode description of the protocol. We discuss its correctness in more depth in Section \ref{sec:correctness}.

\subsection{High Level Protocol Behavior}
\label{sec:protocol-behavior}
At a high level, dynamic reconfiguration in \textit{MongoRaftReconfig} consists of two main aspects: (1) updating the current configuration and (2) propagating new configurations between servers. Configurations also have an impact on election behavior which we discuss below, in Section \ref{sec:config-and-elections}. Formally, a \textit{configuration} is defined as a tuple $(m,v,t)$, where $m \in 2^{Server}$ is a member set, $v \in \mathbb{N}$ is a numeric configuration \textit{version}, and $t \in \mathbb{N}$ is the numeric \textit{term} of the configuration. For convenience, we refer to the elements of a configuration tuple $C=(m,v,t)$ as, respectively, $C.m$, $C.v$ and $C.t$. Each server of a replica set maintains a single, durable configuration, and it is assumed that, initially, all nodes begin with a common configuration, $(m_{init}, 1, 0)$, where $m_{init} \in (2^{Server} \setminus \emptyset)$.
% Note that the set $Server$ for a replica set is fixed in our model of the protocol. Reconfiguration does not change the set of elements in $Server$. Rather, it changes what subset of servers are actively participating in the protocol. 

% The elements of a configuration's member set determine the subset of servers that actively participate in the protocol. For example, if $\{s_1,s_2,s_3,s_4,s_5\}$ is the set of servers in a replica set, and the current configuration $C$ has a member set $C.m=\{s_1,s_2,s_3\}$, the protocol can operate and make progress without a need to contact nodes $s_4$ or $s_5$.

\renewcommand{\algorithmicprocedure}{\textbf{action:}}
% \begin{minipage}{\dimexpr.5\textwidth-.5\columnsep}

\begin{algorithm}
    \caption{Pseudocode description of \textit{MongoRaftReconfig} reconfiguration specific behavior.}
    \label{alg:mrr-pseudocode}

    % State variables in \textcolor{blue}{blue} represent variables that are not affected by reconfigurations i.e. they are modified only by \textit{MongoStaticRaft}.

    % TODO: Consider defining states and initialization up here, separately?
    % \begin{algorithmic}[0]
    %     \footnotesize
    %     %%% States.
    %     \State \underline{\textbf{State and Initialization}}\\\\

    %     Let $m_{init} \in 2^{Server} \setminus \emptyset$\\
    %     For all $i \in Server:$ \\
    %     \ \ $term[i] \in \mathbb{N}$, initially $0$\\
    %     \ \ $state[i] \in \{Pri, Sec\}$, initially $Sec$\\
    %     \ \ $config[i] \in 2^{Server}$, initially $m_{init}$\\
    %     \ \ $configVersion[i] \in \mathbb{N}$, initially $1$\\
    %     \ \ $configTerm[i] \in \mathbb{N}$, initially $0$\\
    % \end{algorithmic}

    \begin{algorithmic}[0]
        \footnotesize
        \State \underline{\textbf{Definitions}}
        \vspace{3pt}
        \State $C_{(i)} \triangleq (config[i], configVersion[i], configTerm[i])$ 
        \State $C_i > C_j \triangleq (C_i.t > C_j.t) \vee (C_i.t = C_j.t \wedge C_i.v > C_j.v)$ 
        \State $C_i \geq C_j \triangleq (C_i > C_j) \vee ((C_i.v, C_i.t) = (C_j.v,C_j.t))$
        % In anonymous version, author names take up less space so we can spread out a bit without forcing pseudocode onto its own page.
        \ifx\authoranonymous\relax 
            \\
        \fi
        % \State $(Pri,Sec) \triangleq (Primary, Secondary)$ 
        % \\
        % \State $InLog(ind, t, s) \triangleq \E k \in 1..Len(log[s]) : (k = ind \wedge log[s][k] = t)$
        % \State $IsCommitted(ind, t, Q) \triangleq \forall s \in Q : InLog(ind, t, s) \wedge term[s] = t$
        % \State $CommittedAt(t) \triangleq \{c \in committed : c[2] = t \}$\\

        \State $Q1(i) \triangleq$ 
        $\exists Q \in Quorums(config[i]) : \forall j \in Q : (C_{(j)}.v, C_{(j)}.t) = (C_{(i)}.v,C_{(i)}.t)$ \Comment Config Quorum Check
        
        \State $Q2(i) \triangleq$ 
        $\exists Q \in Quorums(config[i]) : \forall j \in Q : term[j] = term[i]$ \Comment Term Quorum Check 

        % \State $P1a(i) \triangleq CommittedAt(term[i])\neq \emptyset$
        % \State $P1b(i,Q) \triangleq \forall c \in CommittedAt(term[i]) : IsCommitted(c[1], term[i], Q)$ all entries committed in $term[i]$ are committed in $Q$
        % \State $P1b(i,Q) \triangleq $ entries committed in terms $ \leq term[i]$ are committed in $Q$
        \State $P1(i) \triangleq$ $\exists Q \in Quorums(config[i]) :$ all entries committed in terms $ \leq term[i]$ are committed in $Q$
        % \Comment Oplog Commitment
    \end{algorithmic}

    % To help squeeze the pseudocode onto the page.
    \vspace{-8pt}

    \begin{multicols}{2}
    \begin{algorithmic}[1]
    \footnotesize

    % TODO: Consider condensing variable names for pseudocode description.
    %%% States.
    \State \underline{\textbf{State and Initialization}}\\
    \vspace{3pt}
    Let $m_{init} \in 2^{Server} \setminus \emptyset$\\
    $\forall i \in Server:$ \\
    \ \ $term[i] \in \mathbb{N}$, initially $0$\\
    \ \ $state[i] \in \{Pri., Sec.\}$, initially $Secondary$\\
    \ \ $config[i] \in 2^{Server}$, initially $m_{init}$\\
    \ \ $configVersion[i] \in \mathbb{N}$, initially $1$\\
    \ \ $configTerm[i] \in \mathbb{N}$, initially $0$\\
    
    % \ \ \textcolor{blue}{$log[i] \in Seq(\mathbb{N})$, initially $\langle \rangle$}\\

    % Defining 'committed' set as containing tuples (i,t,term), where the committed entry is (i,t) and
    % it is committed in 'term'.
    % \ \ \textcolor{blue}{$committed \subseteq \mathbb{N} \times \mathbb{N}$, initially $\emptyset$}\\

    %%% Transitions.
    \State \underline{\textbf{Actions}}
    \vspace{3pt}
    \Procedure{Reconfig}{$i$, $m_{new}$} 
    % \Comment{Primary $i$ executes a reconfiguration.}
    % \State Precondition:
    \State \textbf{require} $state[i] = Primary$
    \State \textbf{require} $Q1(i) \wedge Q2(i) \wedge P1(i)$
    \State \textbf{require} $QuorumsOverlap(config[i], m_{new})$
    % \State Effect:
    \State $config[i] \gets m_{new}$
    \State $configVersion[i] \gets configVersion[i] + 1$

    \EndProcedure
    \\
    
    \Procedure{SendConfig}{$i,j$} 
    % \Comment{Server $i$ sends its newer config to server $j$. }
    % \State Precondition:
    \State \textbf{require} $state[j] = Secondary$
    \State \textbf{require} $C_{(i)} > C_{(j)}$
    % \State Effect:
    \State $C_{(j)} \gets C_{(i)}$
    \EndProcedure
    \\

    \Procedure{BecomeLeader}{$i, Q$} 
    % \Comment{Server $i$ becomes a primary.}
    % \State Precondition:
    \State \textbf{require} $ Q \in Quorums(config[i])$
    \State \textbf{require} $i \in Q$
    \State \textbf{require} $\forall v \in Q : C_{(i)} \geq C_{(v)}$
    \State \textbf{require} $\forall v \in Q : term[i] + 1 > term[v]$
    \State $state[i] \gets Primary$
    \State $state[j] \gets Secondary$, $\forall j \in (Q \setminus \{i\})$
    \State $term[j] \gets term[i] + 1$, $\forall j \in Q$
    \State $configTerm[i] \gets term[i] + 1$
    \EndProcedure
    \\

    \Procedure{UpdateTerms}{$i,j$} 
    % \Comment{Server $i$ sends server $j$ its newer term.}
    % \State Precondition:
    \State \textbf{require}  $term[i] > term[j]$
    % \State Effect:
    \State $state[j] \gets Secondary$
    \State $term[j] \gets term[i]$
    \EndProcedure
    \end{algorithmic}
    \end{multicols}
\end{algorithm}

To update the current configuration of a replica set, a client issues a \textit{reconfiguration} command to a primary server with a new, desired configuration, $C'$. Reconfigurations can only be executed on primary servers, and they update the primary's current local configuration $C$ to the specified configuration $C'$. The version of the new configuration, $C'.v$, must be greater than the version of the primary's current configuration, $C.v$, and the term of $C'$ is set equal to the current term of the primary processing the reconfiguration. 
% That is, a reconfiguration from $C_{old}$ to $C_{new}$ on a primary in term $T$ satisfies the conditions $C_{new}.v > C_{old}.v$ and $C_{new}.t = T$.
%%%
After a reconfiguration has occurred on a primary, the updated configuration needs to be communicated to other servers in the replica set. This is achieved in a simple, gossip like manner. Secondaries receive information about the configurations of other servers via periodic heartbeats. They need to have some mechanism, however, for determining whether one configuration is newer than another. This is achieved by totally ordering configurations by their $(version,term)$ pair, where term is compared first, followed by version. If configuration $C_j$ compares as greater than configuration $C_i$ based on this ordering, we say that $C_j$ is \textit{newer} than $C_i$. 
%For two configurations, $C_i$ and $C_j$, if $(C_j.v, C_j.t)$ compares as greater than $(C_i.v, C_i.t)$, we say that $C_j$ is \textit{newer} than $C_i$. 
A secondary can update its configuration to any that is newer than its current configuration. If it learns that another server has a newer configuration, it will fetch that server's configuration, verify that it is still newer than its own upon receipt, and install it locally. 

The above provides a basic outline of how reconfigurations occur and how configurations are propagated between servers in \textit{MongoRaftReconfig}. The pseudocode given in Algorithm \ref{alg:mrr-pseudocode} gives a more abstract and precise description of these behaviors. Note that, in order for the protocol to operate safely, there are several additional restrictions that are imposed on both reconfigurations and elections, which we discuss in more detail below, in Sections \ref{sec:reconfig-safety-restrictions} and \ref{sec:config-and-elections}.

\subsection{Safety Restrictions on Reconfigurations}
\label{sec:reconfig-safety-restrictions}

% To satisfy this, it is sufficient to enforce a \textit{single node change} condition, which requires that no more than a single member is added or removed in a single reconfiguration.

In \textit{MongoStaticRaft}, which does not allow reconfiguration, the safety of the protocol depends on the fact that the \textit{quorum overlap} condition is satisfied for the member sets of any two configurations. This holds since there is a single, uniform configuration that is never modified. For any pair of arbitrary configurations, however, their member sets may not satisfy this property. So, in order for \textit{MongoRaftReconfig} to operate safely, extra restrictions are needed on how nodes are allowed to move between configurations. First, any reconfiguration that moves from $C$ to $C'$ is required to satisfy the quorum overlap condition i.e. $QuorumsOverlap(C.m, C'.m)$. This restriction is discussed in Raft's approach to reconfiguration \cite{OngaroDissertation2014}, and is adopted by \textit{MongoRaftReconfig}. Even if quorum overlap is ensured between any two adjacent configurations, it may not be ensured between all configurations that the system passes through over time. So, there are additional preconditions that must be satisfied before a primary server in term $T$ can execute a reconfiguration out of its current configuration $C$: 
% The sufficiency of this condition to enforce quorum overlap is illustrated in Formula \ref{eq:single-node-change-overlap}.
\begin{enumerate}
\item[Q1.] \textit{Config Quorum Check}: There must be a quorum of servers in $C.m$ that are currently in configuration $C$.
\item[Q2.] \textit{Term Quorum Check}: There must be a quorum of servers in $C.m$ that are currently in term $T$.
\item[P1.] \textit{Oplog Commitment}: All oplog entries committed in terms $\leq T$ must be committed on some quorum of servers in $C.m$.
\end{enumerate}
The above preconditions are stated in Algorithm \ref{alg:mrr-pseudocode} as $Q1(i)$, $Q2(i)$, and $P1(i)$, and they collectively enforce two fundamental requirements needed for safe reconfiguration: \textit{deactivation} of old configurations and \textit{state transfer} from old configurations to new configurations. Q1, when coupled with the election restrictions discussed in Section \ref{sec:config-and-elections}, achieves deactivation by ensuring that configurations earlier than $C$ can no longer elect a primary. Q2 ensures that term information from older configurations is correctly propagated to newer configurations, while P1 ensures that previously committed oplog entries are properly transferred to the current configuration, ensuring that any primary in a current or later configuration will contain these entries. 

% \mytodo{Also mention how we need deactivation of old terms from commits?}
% \mytodo{Reference that correctness section discusses these intuitive arguments in more depth?}

% that committed oplog entries from any previous configurations are now committed by the rules of the current configuration. 

\subsection{Configurations and Elections}
\label{sec:config-and-elections}

When a node runs for election in \textit{MongoStaticRaft}, it must ensure its log is appropriately up to date and that it can garner a quorum of votes in its term. In \textit{MongoRaftReconfig}, there is an additional restriction on voting behavior that depends on configuration ordering. 
If a replica set server is a candidate for election in configuration $C_i$, then a prospective voter in configuration $C_j$ may only cast a vote for the candidate if $C_i$ is newer than or equal to $C_j$.
% if $(v_c, t_c) \geq (v, t)$. 
Furthermore, when a node wins an election, it must update its current configuration with its new term before it is allowed to execute subsequent reconfigurations. That is, if a node with current configuration $(m,v,t)$ wins election in term $t'$, it will update its configuration to $(m,v,t')$ before allowing any reconfigurations to be processed. This behavior is necessary to appropriately deactivate concurrent reconfigurations that may occur on primaries in a different term. This configuration re-writing behavior is analogous to the write in Raft's corrected membership change protocol proposed in \cite{ongaro2015-membership-bug}.

\subsection{Formal Specification}
\label{sec:formal-spec}

The complete, formal description of \textit{MongoRaftReconfig} is given in the TLA+ specification in the supplementary materials \cite{supp-materials}. Note that TLA+ does not impose an underlying system or communication model (e.g. message passing, shared memory), which allows one to write specifications at a wide range of abstraction levels. Our specifications are written at a deliberately high level of abstraction, ignoring some lower level details of the protocol and system model. In practice, we have found the abstraction level of our specifications most useful for understanding and communicating the essential behaviors and safety characteristics of the protocol, while also serving to make automated verification via model checking more feasible, which is discussed further in Section \ref{sec:model-checking-section}. 

% \mytodo{Consider discussing compositional definition of protocols, which allows for efficient verification of separate properties which we want to discuss later. If necessary, possibly put subprotocol refinement proof into appendix.}

% Complete formal specification is provided elsewhere, and a high level description of it is provided in the appendix.

% , which we examine further in section \ref{model-checking-section}

% \subsection{Summary}
% The above provides a high level description of the behaviors of \textit{MongoRaftReconfig} and how it operates safely. In the following section we present our formal specification of the protocol in TLA+, which allows us to define the protocol and its safety properties precisely. Additionally, it allows for automated verification of the protocol's correctness, which we discuss in Section \ref{sec:correctness-analysis}.

\section{Correctness}
\label{sec:correctness}

% \subsection{Safety Properties}

% The fundamental safety property of MongoDB’s core replication protocol, \textit{MongoStaticRaft}, is the \textit{StateMachineSafety} property, which states that if an oplog entry has been marked committed at a particular log index, no conflicting log entry will ever be marked committed at the same index. This is the same high level safety property of the standard Raft protocol. We formally state this property as a predicate on the \textit{committed} variable, which stores the set of committed log entries as $(index, term)$ pairs:
% \begin{align*}
%     StateMachineSafety \triangleq 
%     \forall e_i, e_j \in committed : (e_i[1] = e_j[1]) \Rightarrow (e_i = e_j)
% \end{align*}
% We want to verify that \textit{MongoRaftReconfig} satisfies the same property. This property is an invariant, meaning that all reachable states of the protocol must satisfy it. Thus, our main safety verification goal can be stated formally as:
% \begin{align*}
%     MongoRaftReconfig \Rightarrow \square StateMachineSafety
%     % \label{msr-state-machine-safety}
% \end{align*}

In this section we present a brief outline of our safety proof for \textit{MongoRaftReconfig}. We do not address liveness properties in this work. 
\ifextended
The proof details are left to Appendix \ref{appendix:safety-proof}. 
\else
The full proof is left to \cite{schultz2021design}.
\fi

The key, high level safety property of \textit{MongoRaftReconfig} that we establish in this paper is \textit{LeaderCompleteness}, which is a fundamental safety property of both standard Raft and \textit{MongoStaticRaft}, and is stated below as Theorem \ref{thm:leader-completeness}. This property states that if a log entry has been committed in term $T$, then it must be present in the logs of all primary servers in terms $> T$. Essentially, it ensures that writes committed by some primary will be permanently durable in the replica set. Below we give a high level, intuitive outline of the proof.

\subsection{Overview}

Conceptually, \textit{MongoRaftReconfig} can be viewed as an extension of the \textit{MongoStaticRaft} replication protocol that allows for dynamic reconfiguration. \textit{MongoRaftReconfig}, however, violates the property that all quorums of any two configurations overlap, which \textit{MongoStaticRaft} relies on for safety. It is therefore necessary to examine how \textit{MongoRaftReconfig} operates safely even though it cannot rely on the quorum overlap property. In \textit{MongoStaticRaft}, there are two key aspects of protocol behavior that depend on quorum overlap: (1) \textit{elections} of primary servers and (2) \textit{commitment} of log entries. Elections must ensure that there is at most one unique primary per term, referred to as the \textit{ElectionSafety} property. 
% , referred to as the \emph{ElectionSafety} property, is primarily concerned with ensuring that there is at most one unique leader elected per term. 
Additionally, if a log entry is committed in a given term, it must be present in the logs of all primary servers in higher terms, referred to as the \emph{LeaderCompleteness} property. Both of these safety properties must be upheld in \textit{MongoRaftReconfig}.

% The second, referred to as \emph{Leader Completeness}, is concerned with ensuring that writes committed in a given term are present in the logs of all primary servers in higher terms. 

\emph{LeaderCompleteness} is the essential, high level safety property that we must establish for \textit{MongoRaftReconfig}. \textit{ElectionSafety} is a key, auxiliary lemma that is required in order to show \textit{LeaderCompleteness}. So, this guides the general structure of our proof. Section \ref{sec:elecsafety-proof-outline} presents an intuitive outline of the \textit{ElectionSafety} proof, followed by a similar discussion of \textit{LeaderCompleteness} in Section \ref{sec:leadercomp-proof-outline}. 
\ifextended
The full proofs are left to Appendix \ref{appendix:safety-proof}.
\else
The full proofs are left to \cite{schultz2021design}.
\fi

\subsection{Election Safety}
\label{sec:elecsafety-proof-outline}

In \textit{MongoStaticRaft}, if an election has occurred in term $T$ it ensures that some quorum of servers have terms $\geq T$. This prevents any future candidate from being elected in term $T$, since the quorum required for any future election will contain at least one of these servers, preventing a successful election in term $T$. This property, referred to as \emph{ElectionSafety}, is stated below as Lemma \ref{lemma:election-safety}.
\begin{lemma}[Election Safety]
  \label{lemma:election-safety}
  For all $s,t \in Server$ such that $s \neq t$, it is not the case that both $s$ and $t$ are primary and have the same term.
  \begin{align*}
    \forall s,t &\in Server : \\&(state[s]=Primary \wedge state[t]=Primary \wedge term[s]=term[t]) \Rightarrow (s=t)
  \end{align*} 
\end{lemma}
In \textit{MongoRaftReconfig}, ensuring that a quorum of nodes have terms $\geq T$ after an election in term $T$ is not sufficient to ensure that \textit{ElectionSafety} holds, since there is no guarantee that all quorums of future configurations will overlap with those of past configurations. To address this, \textit{MongoRaftReconfig} must appropriately \textit{deactivate} past configurations before creating new configurations. Conceptually, configurations in the protocol can be considered as either or \textit{active} or \textit{deactivated}, the former being any configuration that is not deactivated. Deactivated configurations cannot elect a new leader or execute a reconfiguration. \textit{MongoRaftReconfig} ensures proper deactivation of configurations by upholding an invariant that the quorums of all active configurations overlap with each other.
% For example, two configurations that do not overlap, like $\{n_1,n_2\}$ and $\{n_3,n_4\}$, have, in general, no way of communicating information between each other, so if they are both in existence and able to conduct elections simultaneously it may lead to a split-brain scenario, where both can carry out operations independently from each other. To prevent this, configurations are ordered by their $(version, term)$ pair, and older configs are deactivated before newer configs are created. 
In addition to deactivation of configurations, \emph{MongoRaftReconfig} must also ensure that term information from one configuration is properly transferred to subsequent configurations, so that later configurations know about elections that occurred in earlier configurations. For example, if an election occurred in term $T$ in configuration $C$, even if $C$ is deactivated by the time $C'$ is created, the protocol must also ensure that $C'$ is ``aware'' of the fact that an election in $T$ occurred in $C$. \emph{MongoRaftReconfig} ensures this by upholding an additional invariant stating that the quorums of all active configurations overlap with some server in term $\geq T$, for any past election that occurred in term $T$. 

Collectively, the two above invariants are the essential properties for understanding how the \emph{ElectionSafety} property is upheld in \emph{MongoRaftReconfig}. 
\ifextended
The formal statement of these invariants and the complete proof is left to Appendix \ref{appendix:elecsafety-proof}.
\else
The formal statement of these invariants and the complete proof is left to \cite{schultz2021design}.
\fi
In the following section, we briefly discuss the \textit{LeaderCompleteness} property and its proof, which relies on the \emph{ElectionSafety} property. 

\subsection{Leader Completeness}
\label{sec:leadercomp-proof-outline}

\emph{LeaderCompleteness} is the key high level safety property of \textit{MongoRaftReconfig}. It ensures that if a log entry is committed in term $T$, then it is present in the logs of all leaders in terms $ > T$. Essentially, it ensures that committed log entries are durable in a replica set. It is stated below as Theorem \ref{thm:leader-completeness}, where $committed \in \mathbb{N} \times \mathbb{N}$ refers to the set of committed log entries as $(index, term)$ pairs, and $InLog(i,t,s)$ is a predicate determining whether a log entry $(i,t)$ is contained in the log of server $s$. 
\begin{theorem}[Leader Completeness]
  \label{thm:leader-completeness}
  If a log entry is committed in term $T$, then it is present in the log of any leader in term $T' > T$.
    \begin{align}
      \begin{split}
        \A s \in &Server : \A (cindex, cterm) \in committed : \\
        &( state[s] = Primary \wedge cterm < term[s]) \Rightarrow InLog(cindex, cterm, s)
      \end{split}
    \end{align}
\end{theorem}
In \textit{MongoStaticRaft}, \emph{LeaderCompleteness} is ensured due to the overlap between quorums used for commitment of a write and quorums used for the election of a primary. In \textit{MongoRaftReconfig}, this does not hold, so the protocol instead upholds a more general invariant, stating that, for all committed entries $E$, the quorums of all active configurations overlap with some server that contains $E$ in its log. \textit{MongoRaftReconfig} also ensures that newer configurations appropriately disable commitment of log entries in older terms. 
\ifextended
We defer the statement of these invariants and the complete proof of Theorem \ref{thm:leader-completeness} and its supporting lemmas to Appendix \ref{appendix:leader-comp-proof}.
\else
We defer the statement of these invariants and the complete proof of Theorem \ref{thm:leader-completeness} and its supporting lemmas to \cite{schultz2021design}.
\fi

\subsection{Model Checking} 
\label{sec:model-checking-section}

% \mytodo{Revisit paragraph breaks in this section.}

In addition to the safety proof outlined above, we used TLC \cite{yu1999model}, an explicit state model checker for TLA+ specifications, to gain additional confidence in the safety of the protocol. We consider it important to augment the human reasoning process for protocols like this with some type of machine based verification, even if the verification is incomplete, since it is easy for humans to make subtle errors in reasoning when considering distributed protocols of this nature. 
% In addition, TLC is able to provide concrete counterexamples to the user if it discovers a correctness property violation.
% It has been observed elsewhere \cite{Ma2019} that relatively small, finite instances of distributed protocols are often sufficient to exhibit behaviors that are generalizable to larger (potentially infinite) instances, which helps to provide confidence in our approach. 

We verified fixed, finite instances of \textit{MongoRaftReconfig} to provide a sound guarantee of protocol correctness for given parameters. \textit{MongoRaftReconfig} is an infinite state protocol, so verification via explicit state model checking is, necessarily, incomplete. That is, it does not establish correctness of the protocol for an unbounded number of servers or system parameters. It does, however, provide a strong initial level of confidence that the protocol is safe. A goal for future work is to develop a complete, machine checked safety proof using the TLA+ proof system \cite{chaudhuri2010verifying}.

% and that there are no obvious errors in the design for suitably sized system instances
\subsubsection{Methodology and Results}
Formally, \textit{MongoRaftReconfig} behaves as an extension of \textit{MongoStaticRaft} that allows for dynamic reconfiguration. Thus, it can be viewed as a composition of two distinct subprotocols: one for managing the oplog, and one for managing configuration state. The oplog is maintained by  \textit{MongoStaticRaft}, and configurations are maintained by a protocol we refer to below as \textit{MongoLoglessDynamicRaft}, which implements the logless replicated state machine that manages the configuration state of the replica set. Algorithm \ref{alg:mrr-pseudocode} summarizes the behaviors of \textit{MongoLoglessDynamicRaft}. This compositional approach to describing \textit{MongoRaftReconfig} is formalized in our TLA+ specification which can be found in the supplementary materials \cite{supp-materials}. Our verification efforts centered on checking the two key safety properties discussed in the above sections, \textit{ElectionSafety} and \textit{LeaderCompleteness}. We summarize the results below,
\ifextended
leaving the full details to Appendix \ref{appendix:model-checking-details}.
\else
leaving the full details to \cite{schultz2021design}.
\fi

% We automatically verified the \textit{StateMachineSafety} invariant by model checking a finite instance of the \textit{MongoRaftReconfig} specification. 

\subparagraph{Checking Leader Completeness}
We were able to successfully verify the \textit{LeaderCompleteness} property on a finite instance of \textit{MongoRaftReconfig} with 4 servers, logs of maximum length 2, maximum configuration versions of 3, and maximum server terms of 3. That is, we manually imposed a constraint preventing the model checker from exploring any states exceeding these finite bounds. Model checking this instance generated approximately 345 million distinct protocol states and took approximately 8 hours to complete with 20 TLC worker threads on a 48-core, 2.30GHz Intel Xeon Gold 5118 CPU. 

\subparagraph{Checking Election Safety}

As evidenced by the above metrics, it was difficult to scale verification of the \emph{LeaderCompleteness} property to much larger system parameters. So, to provide additional confidence, we checked the \emph{ElectionSafety} property on the \textit{MongoLoglessDynamicRaft} protocol in isolation, which allowed us to verify instances with significantly larger parameters. Due to the compositional structure of \textit{MongoRaftReconfig}, verifying that the \emph{ElectionSafety} property holds on \textit{MongoLoglessDynamicRaft} is sufficient to ensure that it holds in \textit{MongoRaftReconfig}. 
Intuitively, the additional preconditions imposed by \textit{MongoRaftReconfig} only \textit{restrict} the behaviors of \textit{MongoLoglessDynamicRaft}, but do not augment them. 
\ifextended
We formalize and prove this fact via a refinement based argument, whose details are left to Appendix \ref{sec:subprotocol-refinement}.
\else
We formalize and prove this fact via a refinement based argument, whose details are left to \cite{schultz2021design}.
\fi
This allows us to assume our verification efforts for \textit{MongoLoglessDynamicRaft} hold in \textit{MongoRaftReconfig}, providing stronger confidence in the correctness of the overall protocol.

We successfully verified the \emph{ElectionSafety} property on a finite instance of \textit{MongoLoglessDynamicRaft} with 5 servers, maximum configuration versions of 4, and maximum terms of 4. Model checking this instance generated approximately 812 million distinct states and took around 19.5 hours to complete with 20 TLC worker threads on a 48-core, 2.30GHz Intel Xeon Gold 5118 CPU. The ability to check these considerably larger parameter values in only several extra hours of wall clock time demonstrates the effectiveness of this compositional model checking approach, helping us mitigate state space explosion \cite{clarke2011model}.

\section{Conceptual Insights}
\label{sec:raft-comparison}

% \section{Comparison to Standard Raft}
% \subsection{Behavioral Mapping}

 \textit{MongoRaftReconfig} can be viewed as a generalization and optimization of the standard Raft reconfiguration protocol. To explain the conceptual novelties of our protocol and how it relates to standard Raft, we discuss below the two primary aspects of the protocol which set it apart from Raft: (1) decoupling of the oplog and config state machine and (2) logless optimization of the config state machine. These are covered in Sections \ref{sec:decoupling-reconfig} and \ref{sec:logless-optimization}, respectively. Section \ref{sec:experiments} provides an experimental evaluation of how these novel aspects can provide  performance benefits for reconfiguration, by allowing reconfigurations to bypass the main operation log in cases where it has become slow or stalled.

\subsection{Decoupling Reconfigurations}
\label{sec:decoupling-reconfig}

In standard Raft, the main operation log is used for both normal operations and reconfiguration operations. This coupling between logs has the benefit of providing a single, unified data structure to manage system state, but it also imposes fundamental restrictions on the operation of the two logs. Most importantly, in order for a write to commit in one log, it must commit all previous writes in the other. For example, if a reconfiguration log entry $C_j$ has been written at log index $j$ on primary $s$, and there is a sequence of uncommitted log entries $U=\langle i,i+1,...,j-1 \rangle$ in the log of $s$, in order for a reconfiguration from $C_j$ to $C_k$ to occur, all entries of $U$ must become committed. This behavior, however, is stronger than necessary for safety i.e. it is not strictly necessary to commit these log entries before executing a reconfiguration. The only fundamental requirements are that previously committed log entries are committed by the rules of the current configuration, and that the current configuration has satisfied the necessary safety preconditions. Raft achieves this goal implicitly, but more conservatively than necessary, by committing the entry $C_j$ and all entries behind it. This ensures that all previously committed log entries, in addition to the uncommitted operations $U$, are now committed in $C_j$, but it is not strictly necessary to pipeline a reconfiguration behind commitment of $U$. \textit{MongoRaftReconfig} avoids this by separating the oplog and config state machine and their rules for commitment and reconfiguration,  allowing reconfigurations to bypass the oplog if necessary. Section \ref{sec:experiments} examines this aspect of the protocol experimentally.

\subsection{Logless Optimization}
\label{sec:logless-optimization}

Decoupling the config state machine from the main operation log allows for an optimization that is enabled by the fact that reconfigurations are ``update-only" operations on the replicated state machine. This means that it is sufficient to store only the latest version of the replicated state, since the latest version can be viewed as a ``rolled-up" version of the entire (infinite) log. This logless optimization allows the configuration state machine to avoid complexities related to garbage collection of old log entries and it simplifies the mechanism for state propagation between servers. Normally, log entries are replicated incrementally, either one at a time, or in batches from one server to another. Additionally, servers may need to have an explicit procedure for deleting (i.e. rolling back) log entries that will never become committed. In the logless replicated state machine, all of these mechanisms can be combined into a single conceptual action, that atomically transfers the entire log of server $s$ to another server $t$, if the log of $s$ is newer, based on the index and term of its last entry. In \textit{MongoRaftReconfig}, this is implemented by the \textit{SendConfig} action, which transfers configuration state from one server to another.

% we refer to as \textit{MergeEntries}. This action conceptually subsumes the \textit{GetEntries} and \textit{RollbackEntries} actions of the \textit{MongoStaticRaft} specification described in Section \ref{sec:msr-formal-spec}, which is a log-based protocol (\mytodo{Reword this section reference}). We do not formally define \textit{MergeEntries} here, but it can be viewed as an action where one server $s$ atomically transfers its entire log to another server $t$, if the log of $s$ is newer, based on the index and term of its last entry. 

%
% TODO: Move experiments to its own section???
%
\section{Experimental Evaluation}
\label{sec:experiments}

In a healthy replica set, it is possible that a failure event causes some subset of replica set servers to degrade in performance, causing the main oplog replication channel to become lagged or stall entirely. If this occurs on a majority of nodes, then the replica set will be prevented from committing new writes until the performance degradation is resolved. For example, consider a 3 node replica set consisting of nodes $\{n0, n1,n2\}$, where nodes $n1$ and $n2$ suddenly become slow or stall replication. An operator or failure detection module may want to reconfigure these nodes out of the set and add in two new, healthy nodes, $n3$ and $n4$, so that the system can return to a healthy operational state.  This requires a series of two reconfigurations, one to add $n3$ and one to add $n4$. In standard Raft, this would require the ability to commit at least one reconfiguration oplog entry with one of the degraded nodes ($n1$ or $n2$). This prevents such a reconfiguration until the degradation is resolved. In \textit{MongoRaftReconfig}, reconfigurations bypass the oplog replication channel, committing without the need to commit writes in the oplog. This allows \textit{MongoRaftReconfig} to successfully reconfigure the system in such a degraded state, restoring oplog write availability by removing the failed nodes and adding in new, healthy nodes. 

%
% Paragraph in response to PODC Reviewer 60C.
%
Note that if a replica set server experiences a period of degradation (e.g. a slow disk), both the oplog and reconfiguration channels will be affected, which would seem to nullify the benefits of decoupling the reconfiguration and oplog replication channels. In practice, however, the operations handled by the oplog are likely orders of magnitude more resource intensive than reconfigurations, which typically involve writing a negligible amount of data. So, even on a degraded server, reconfigurations should be able to complete successfully when more intensive oplog operations  become prohibitively slow, since the resource requirements of reconfigurations are extremely lightweight.

% TODO: Even in cases where a slow machine has, for example, degraded disk performance, this still offers a potential benefit since the oplog generally processes operations that maybe be much more expensive than a reconfiguration write, which involves an extremely small amount of data. For example, the oplog may be inserting or updating a large document, which could stall this channel with a slow disk/machine. Reconfigs presumably could still succeed, though, since their resource requirements in terms of disk/CPU are essentially negligible.

\subsection{Experiment Setup and Operation}

To demonstrate the benefits of \textit{MongoRaftReconfig} in this type of scenario, we designed an experiment to measure how quickly a replica set can reconfigure in new nodes to restore majority write availability when it faces periodic phases of degradation. For comparison, we implemented a simulated version of the Raft reconfiguration algorithm in MongoDB by having reconfigurations write a no-op oplog entry and requiring it to become committed before the reconfiguration can complete \cite{mongo-experiment-repo}. Our experiment initiates a 5 node replica set with servers we refer to as $\{n0,n1,n2,n3,n4\}$. We run the server processes co-located on a single Amazon EC2 t2.xlarge instance with 4 vCPU cores, 16GB memory, and a 100GB EBS disk volume, running Ubuntu 20.04. Co-location of the server processes is acceptable since the workload of the experiment does not saturate any resource (e.g. CPU, disk) of the machine. The servers run MongoDB version v4.4-39f10d with a patch to fix a minor bug \cite{jira-server-46907} that prevents optimal configuration propagation speed in some cases.

\begin{figure}
    \centering
    \includegraphics[scale=0.98]{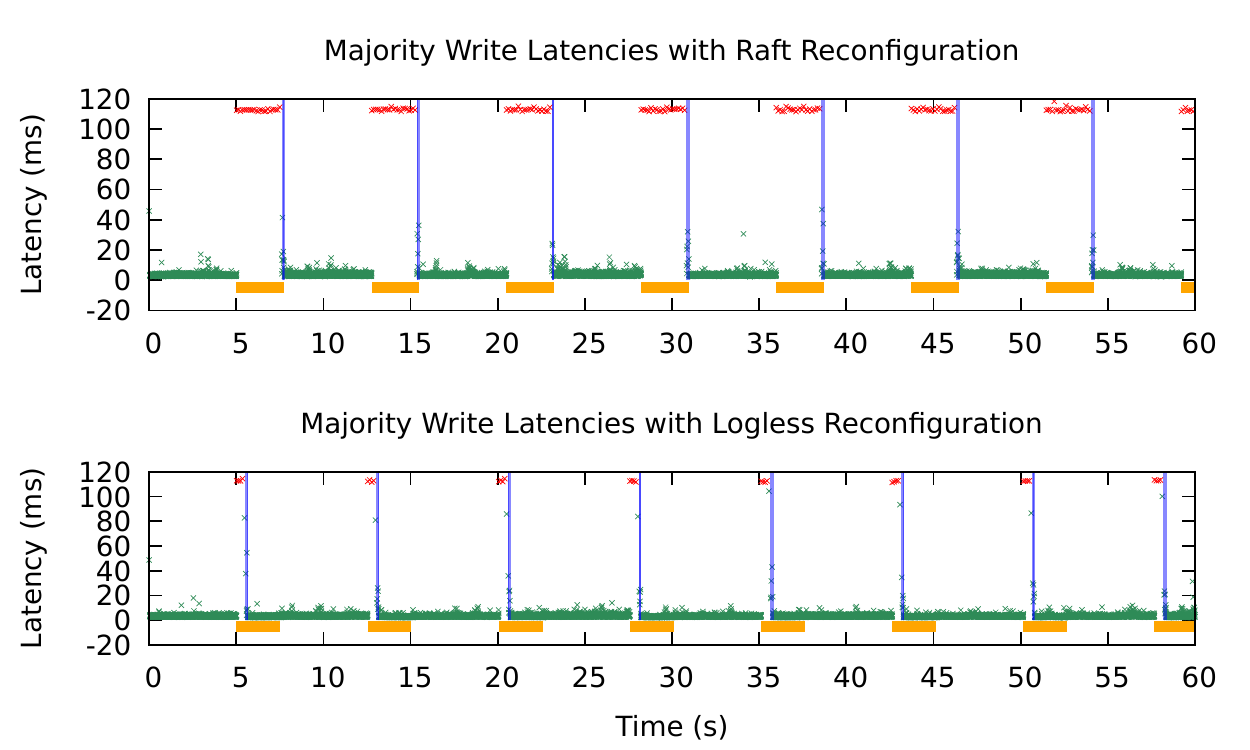}
    \caption{Latency of majority writes in the face of node degradation and reconfiguration to recover. Red points indicate writes that timed out i.e. failed to commit. Orange horizontal bars indicate intervals of time where system entered a \textit{degraded} mode. Thin, vertical blue bars indicate successful completion of reconfiguration events.}
    \label{fig:write-latencies-under-reconfig}
\end{figure}

Initially, $\{n0,n1,n2\}$ are voting servers and $\{n3,n4\}$ are non voting. In a MongoDB replica set, a server can be assigned either 0 or 1 votes. A non-voting server has zero votes and it does not contribute to a commit majority i.e. it is not considered as a member of the consensus group.
% e.g. even if a write is replicated to a majority of non-voting servers it is not marked as committed.  
Our experiment has a single writer thread that continuously inserts small documents into a collection with write concern \textit{majority}, with a write concern timeout of 100 milliseconds. There is a concurrent fault injector thread that periodically simulates a degradation of performance on two secondary nodes by temporarily pausing oplog replication on those nodes. This thread alternates between \textit{steady} periods and \textit{degraded} periods of time, starting out in \textit{steady} mode, where all nodes are operating normally. It runs for 5 seconds in \textit{steady} mode, then transitions to \textit{degraded} mode for 2.5 seconds, before transitioning back to \textit{steady} mode and repeating this cycle. When the fault injector enters \textit{degraded} mode, the main test thread simulates a ``fault detection" scenario (assuming some external module detected the performance degradation) by sleeping for 500 milliseconds, and then starting a series of reconfigurations to add two new, healthy secondaries and remove the two degraded secondaries. Over the course of the experiment, which has a 1 minute duration, we measure the latency of each operation executed by the writer thread. These latencies are depicted in the graphs of Figure \ref{fig:write-latencies-under-reconfig}. Red points indicate writes that failed to commit i.e. that timed out at 100 milliseconds. The successful completion of reconfigurations are depicted with vertical blue bars. It can be seen how, when a period of degradation begins, the logless reconfiguration protocol is able to complete a series of reconfigurations quickly to get the system back to a healthy state, where writes are able to commit again and latencies drop back to their normal levels. In the case of Raft reconfiguration, writes continue failing until the period of degradation ends, since the reconfigurations to add in new healthy nodes cannot complete.

\section{Related Work}

Dynamic reconfiguration in consensus based systems has been explored from a variety of perspectives for Paxos based systems. In Lamport's presentation of Paxos \cite{Lamport2001}, he suggests using a fixed parameter $\alpha$ such that the configuration for a consensus instance $i$ is governed by the configuration at instance $i-\alpha$. This restricts the number of commands that can be executed until the new configuration becomes committed, since the system cannot execute instance $i$ until it knows what configuration to use, potentially causing availability issues if reconfigurations are slow to commit. Stoppable Paxos \cite{malkhi2008stoppable} was an alternative method later proposed where a Paxos system can be reconfigured by stopping the current state machine and starting up a new instance of the state machine with a potentially different configuration. This ``stop-the-world" approach can hurt availability of the system while a reconfiguration is being processed. Vertical Paxos allows a Paxos state machine to be reconfigured in the middle of reaching agreement, but it assumes the existence of an external configuration master \cite{vertical-paxos}. In \cite{paxosmadelive}, the authors describe the Paxos implementation underlying Google's Chubby lock service, but do not include details of their approach to dynamic reconfiguration, stating that ``While group membership with the core Paxos algorithm is straightforward, the exact details are non-trivial when we introduce Multi-Paxos...". They remark that the details, though minor, are ``...subtle and beyond the scope of this paper". 

The Raft consensus protocol, published in 2014 by Ongaro and Ousterhout \cite{raftpaper}, presented two methods for dynamic membership changes: single server membership change and joint consensus. A correctness proof of the core Raft protocol, excluding dynamic reconfiguration, was included in Ongaro's PhD dissertation \cite{OngaroDissertation2014}. Formal verification of Raft's linearizability guarantees was later completed in Verdi \cite{Woos2016}, a framework for verifying distributed systems in the Coq proof assistant \cite{bertot2013interactive}, but formalization of dynamic reconfiguration was not included. In 2015, after Raft's initial publication, a safety bug in the single server reconfiguration approach was found by Amos and Zhang \cite{Amos}, at the time PhD students working on a project to formalize parts of Raft's original reconfiguration algorithm. A fix was proposed shortly after by Ongaro \cite{ongaro2015-membership-bug}, but the project was never extended to include the fixed version of the protocol. The Zab replication protocol, implemented in Apache Zookeeper \cite{Shraer2019}, also includes a dynamic reconfiguration approach for primary-backup clusters that is similar in nature to Raft's joint consensus approach. 

The concept of decoupling reconfiguration from the main data replication channel has previously appeared in other replication systems, but none that integrate with a Raft-based system.  RAMBO \cite{Gilbert2010}, an algorithm for implementing a distributed shared memory service, implements a dynamic reconfiguration module that is loosely coupled with the main read-write functionality. Additionally,  Matchmaker Paxos \cite{whittaker2020matchmaker} is a more recent approach for reconfiguration in Paxos based protocols that adds dedicated nodes for managing reconfigurations, which decouples reconfiguration from the main processing path, preventing performance degradation during configuration changes. There has also been prior work on reconfiguration using weaker models than consensus  \cite{Jehl2014}, and approaches to logless implementations of Paxos based replicated state machine protocols \cite{Rystsov2018}, which bear conceptual similarities to our logless protocol for managing configuration state. Similarly, \cite{kuznetsov_et_al} presents an approach to asynchronous reconfiguration under a Byzantine fault model that avoids reaching consensus on configurations by utilizing a lattice agreement abstraction.

% Our formal specification and verification efforts follow prior lines of work on formally verifying distributed protocols e.g. Paxos and its variants \cite{chand2016formal,lamport2011byzantizing}, the Chord \cite{lund2019verification} protocol, the Pastry distributed hash table \cite{pastry-verify}, and others \cite{braithwaite2020formal,newcombe2014use}.
% Distributed protocols are subtle and challenging to design correctly, so they  benefit greatly from precise, machine checkable descriptions. More recent progress has also been made on tools to help automate the verification and proof of protocols like these even further e.g. Ivy \cite{padon2016ivy} and I4 \cite{Ma2019}. 

\section{Conclusions and Future Work}

In this paper we presented \textit{MongoRaftReconfig}, a novel, logless dynamic reconfiguration protocol that improves upon and generalizes the single server reconfiguration protocol of standard Raft by decoupling the main operation and reconfiguration logs. Although \textit{MongoRaftReconfig} was developed for and presented in the context of the MongoDB system, the ideas and underlying protocol generalize to other Raft-based replication protocols that require dynamic reconfiguration. Goals for future work include development of a machine checked safety proof of the protocol's correctness with help of the TLA+ proof system \cite{chaudhuri2010verifying}, in addition to running more in depth experiments to evaluate how \textit{MongoRaftReconfig} behaves under more varied workloads.

\bibliography{library}

\appendix

\ifextended

\section{Detailed Safety Proof}
\label{appendix:safety-proof}
In this section we provide the detailed proof of Theorem \ref{thm:leader-completeness}, the \textit{LeaderCompleteness} safety property of \emph{MongoRaftReconfig}. We first provide some preliminary definitions in Section \ref{appendix:prelim-defs} that are used throughout the proof. Section \ref{appendix:elecsafety-proof} covers the proof of the \emph{ElectionSafety} property, Section \ref{appendix:log-props} proves some auxiliary properties about logs of servers in the system, and Section \ref{appendix:leader-comp-proof} presents the proof of Theorem \ref{thm:leader-completeness}, which relies on the lemmas established in the preceding sections.

Along with the proof, we provide a complete pseudocode description of \textit{MongoRaftReconfig} in Algorithm \ref{alg:mrr-pseudocode-full}, which is necessary to fully understand and verify the proof. Algorithm \ref{alg:mrr-pseudocode-full} includes the reconfiguration specific behaviors shown in Algorithm \ref{alg:mrr-pseudocode} along with the behaviors of \textit{MongoStaticRaft}, which are mostly orthogonal to reconfiguration, but are necessary to describe for completeness.

\subsection{Definitions and Notation}
\label{appendix:prelim-defs}

Recall that, in our system model, we consider a set of processes $Server=\{s_1,s_2,...,s_n\}$ that communicate by sending messages. Also, recall that a \textit{configuration} is defined as a tuple $(m,v,t)$, where $m \in 2^{Server}$ is a member set, $v \in \mathbb{N}$ is a numeric configuration \textit{version}, and $t \in \mathbb{N}$ is the numeric \textit{term} of the configuration. We refer to the elements of a configuration tuple $C=(m,v,t)$ as, respectively, $C.m$, $C.v$ and $C.t$, and, for a configuration $C$, we informally refer to the elements of $Quorums(C.m)$ as the ``quorums of $C$". For a sequence of elements $L$, we use $L[i]$ to refer to the $i$-th element of $L$ (1-indexed), and $L[..i]$ to refer to the sequence containing the first $i$ elements of $L$. We denote the concatenation of two sequences, $L$ and $M$, as $L \circ M$, a concrete sequence of values as $\langle v_1,v_2,\dots,v_n \rangle$, and the empty sequence as $\langle \rangle$. We also use the notation $1..n$ to refer to the set of natural numbers $\{1,2,\dots,n\}$. 

When referring to state variables of the protocol in inductive proof arguments below, we follow a convention of referring to the value of a state variable $X$ in the current state as $X$ and its value in the next state as $X'$ (i.e. its value after some state transition). When referring to log entries on a server, we sometimes use the pair notation $(index, term)$ to refer to a log entry at position $index$ with a term of $term$. For example, if we say that the log of server $s$ contains entry $(index, term)$, this means that $Len(log[s]) \geq index$ and $log[s][index]=term$. Below we provide a few basic definitions utilized in the proof and in the pseudocode description given in Algorithm \ref{alg:mrr-pseudocode-full}. Additional notation and definitions are also given in Algorithm \ref{alg:mrr-pseudocode-full} and relied upon throughout the proof.

\renewcommand{\algorithmicprocedure}{\textbf{action:}}
% \beingin{minipage}{\dimexpr.5\textwidth-.5\columnsep}

\begin{algorithm}
    \caption{Complete pseudocode description of \textit{MongoRaftReconfig} behavior. State, actions, and behavior specific to \textit{MongoStaticRaft} is highlighted in \textcolor{blue}{blue}.}
    \label{alg:mrr-pseudocode-full}

    % TODO: Consider defining states and initialization up here, separately?
    % \begin{algorithmic}[0]
    %     \footnotesize
    %     %%% States.
    %     \State \underline{\textbf{State and Initialization}}\\\\

    %     Let $m_{init} \in 2^{Server} \setminus \emptyset$\\
    %     For all $i \in Server:$ \\
    %     \ \ $term[i] \in \mathbb{N}$, initially $0$\\
    %     \ \ $state[i] \in \{Pri, Sec\}$, initially $Sec$\\
    %     \ \ $config[i] \in 2^{Server}$, initially $m_{init}$\\
    %     \ \ $configVersion[i] \in \mathbb{N}$, initially $1$\\
    %     \ \ $configTerm[i] \in \mathbb{N}$, initially $0$\\
    % \end{algorithmic}

    \begin{algorithmic}[0]
        \footnotesize
        \State \underline{\textbf{Definitions}}\\

        \State $Seq(S) \triangleq$ the set of all sequences with elements from the set $S$
        \State $Len(s) \triangleq$ the length of a sequence $s$
        
        \State $InLog(ind, t, s) \triangleq \E k \in 1..Len(log[s]) : (k = ind \wedge log[s][k] = t)$
        \State $IsPrefix(l_i, l_j) \triangleq Len(l_i) \leq Len(l_j) \wedge l_i = l_j[..Len(l_i)]$\\

        % \State $LastEntry(i) \triangleq (Len(log[i]), log[Len(log[i])])$
        \State $LogTerm(i) \triangleq$ \algorithmicif\ $log[i] = \langle \rangle$ \algorithmicthen\ $-1$ \algorithmicelse\ $log[i][Len(log[i])]$
        
        % \State $(ind_i, t_i) > (ind_j, t_j) \triangleq (t_i > t_j) \vee (t_i = t_j \wedge ind_i > ind_j)$\\
        % \State $EntryGeq((ind_i, t_i)),(ind_j, t_j)) \triangleq (t_i > t_j) \vee (t_i = t_j \wedge ind_i > ind_j)$\\
        
        \State $LogGeq(i,j) \triangleq (LogTerm(i) > LogTerm(j)) \vee (LogTerm(i) = LogTerm(j) \wedge Len(log[i]) \geq Len(log[j]))$ \label{mrr-full:line-log-geq-def}
        
        % EntryGeq(LastEntry(i), LastEntry(j))$
        % (LogTerm(i) > LogTerm(j)) \vee LogTerm(i) = LogTerm(j)$
        
        % LastEntry(i) > LastEntry(j) \vee LastEntry(i) = LastEntry(j)$

        \State $LogCheck(i,j) \triangleq (Len(log[j])> Len(log[i])) \wedge (log[i] = \langle \rangle \vee (log[i][Len(log[i])] = log[j][Len(log[i])]))$
        \State $CanRollback(i,j) \triangleq (LogTerm(i) < LogTerm(j)) \wedge \neg IsPrefix(log[i],log[j])$
        \State $IsCommitted(ind, t, Q) \triangleq \forall j \in Q : InLog(ind, t, j) \wedge term[j] = t$
        \State $CommittedAt(t) \triangleq \{(index, term) \in committed : term = t \}$\\

        % \State $(Pri,Sec) \triangleq (Primary, Secondary)$\\

        \State $QuorumsAt(i) \triangleq Quorums(config[i])$
        \State $C_{(i)} \triangleq (config[i], configVersion[i], configTerm[i])$\\
        
        % \State $C_i > C_j \triangleq (C_i.t > C_j.t) \vee (C_i.t = C_j.t \wedge C_i.v > C_j.v)$
        % \State $C_i \geq C_j \triangleq (C_i > C_j) \vee ((C_i.v, C_i.t) = (C_j.v,C_j.t))$\\
        
        \State $Q1(i) \triangleq$ 
        $\exists Q \in QuorumsAt(i) : \forall j \in Q : (C_{(j)}.v, C_{(j)}.t) = (C_{(i)}.v, C_{(i)}.t)$ \Comment Config Quorum Check
        
        \State $Q2(i) \triangleq$ 
        $\exists Q \in QuorumsAt(i) : \forall j \in Q : term[j] = term[i]$ \Comment Term Quorum Check \\
        
        \State $P1a(i) \triangleq (committed = \emptyset) \vee CommittedAt(term[i])\neq \emptyset$
        \State $P1b(i,Q) \triangleq \forall c \in CommittedAt(term[i]) : IsCommitted(c[1], term[i], Q)$
        \State $P1(i) \triangleq$ $\exists Q \in QuorumsAt(i) : P1a(i) \wedge P1b(i,Q)$
        \Comment Oplog Commitment

    \end{algorithmic}

    \begin{multicols}{2}
    \begin{algorithmic}[1]
    \footnotesize

    % TODO: Consider condensing variable names for pseudocode description.

    %%% States.
    \State \underline{\textbf{State and Initialization}}\\\\
    
    Let $m_{init} \in 2^{Server} \setminus \emptyset$\\
    $\forall i \in Server:$ \\
    \ \ $term[i] \in \mathbb{N}$, initially $0$\\
    \ \ $state[i] \in \{Pri., Sec.\}$, initially $Secondary$\\
    \ \ $config[i] \in 2^{Server}$, init $m_{init}$\\
    \ \ $configVersion[i] \in \mathbb{N}$, initially $1$\\
    \ \ $configTerm[i] \in \mathbb{N}$, initially $0$\\
    \ \ \textcolor{blue}{$log[i] \in Seq(\mathbb{N})$, initially $\langle \rangle$}\\
    \ \ \textcolor{blue}{$committed \subseteq \mathbb{N} \times \mathbb{N}$, initially $\emptyset$}\\

    %%% Transitions.
    \State \underline{\textbf{Actions}}\\

    \Procedure{Reconfig}{$i$, $m_{new}$} 
    % \Comment{Primary $i$ executes a reconfiguration.}
    % \State Precondition:
    \State \textbf{require} $m_{new} \in 2^{Server}$
    \State \textbf{require} $state[i] = Primary$
    \State \textbf{require} $Q1(i) \wedge Q2(i) \wedge P1(i)$ \label{mrr-full:line-reconfig-precond-quorum-conds}
    \State \textbf{require} $QuorumsOverlap(config[i], m_{new})$ \label{mrr-full:line-reconfig-precond-quorums-overlap}
    % \State Effect:
    \State $config[i] \gets m_{new}$
    \State $configVersion[i] \gets configVersion[i] + 1$
    \EndProcedure
    \\

    \Procedure{SendConfig}{$i,j$} 
    % \Comment{Server $i$ sends its newer config to server $j$. }
    % \State Precondition:
    \State \textbf{require} $state[j] = Secondary$
    \State \textbf{require} $C_{(i)} > C_{(j)}$ \label{mrr-full:line-send-config-precond-gt}
    % \State Effect:
    \State $C_{(j)} \gets C_{(i)}$
    \EndProcedure
    \\

    \Procedure{BecomeLeader}{$i, Q$} 
    % \Comment{Server $i$ becomes a primary.}
    % \State Precondition:
    \State \textbf{require} $ Q \in Quorums(config[i])$
    \State \textbf{require} $i \in Q$
    \State \textbf{require} $\forall v \in Q : C_{(i)} \geq C_{(v)}$ \label{mrr-full:line-become-leader-config-precond}
    \State \textbf{require} $\forall v \in Q : term[i] + 1 > term[v]$ \label{mrr-full:line-become-leader-pre-newer-term}
    \State \textbf{\textcolor{blue}{require}} \textcolor{blue}{$\forall v \in Q : LogGeq(i, v)$}
    \State $state[i] \gets Primary$
    \State $state[j] \gets Secondary$, $\forall j \in (Q \setminus \{i\})$
    \State $term[j] \gets term[i] + 1$, $\forall j \in Q$ \label{mrr-full:line-become-leader-post-term-update}
    \State $configTerm[i] \gets term[i] + 1$
    \EndProcedure
    \\

    \Procedure{UpdateTerms}{$i,j$} 
    % \State Precondition:
    \State \textbf{require}  $term[i] > term[j]$
    % \State Effect:
    \State $state[j] \gets Secondary$
    \State $term[j] \gets term[i]$
    \EndProcedure
    \\

    \Procedure{\textcolor{blue}{ClientRequest}}{$i$} 
    % \State Precondition:
    \State \textbf{require} $state[i] = Primary$
    \State $log[i] \gets log[i] \circ \langle term[i] \rangle$
    \EndProcedure
    \\

    \Procedure{\textcolor{blue}{GetEntries}}{$i,j$} 
    % \State Precondition:
    \State \textbf{require} $state[i] = Secondary$
    \State \textbf{require} $LogCheck(i,j)$
    \State $log[i] \gets log[i] \circ \langle log[j][Len(log[i])+1] \rangle$
    \EndProcedure
    \\

    \Procedure{\textcolor{blue}{RollbackEntries}}{$i,j$} 
    % \State Precondition:
    \State \textbf{require} $state[i] = Secondary$
    \State \textbf{require} $CanRollback(i,j)$
    \State $log[i] \gets log[i][..(Len(log[i])-1)]$
    \EndProcedure
    \\

    \Procedure{\textcolor{blue}{CommitEntry}}{$i,Q$} 
    % \State Precondition:
    \State \textbf{require} $Q \in Quorums(config[i])$
    \State \textbf{require} $state[i] = Primary$
    \State \textbf{require} $IsCommitted(Len(log[i]),term[i], Q)$\label{mrr-full:line-commit-entry-precond}
    % For prefix committed extension.
    % \State $E$ $\gets$ $\{ (k,log[i][k], term[i]) : k \in 1..Len(log[i]) \}$
    % \State $committed \gets committed \cup E$
    %
    \State $committed \gets committed \cup \{ (Len(log[i]), term[i]) \}$
    \EndProcedure
    \\
    \end{algorithmic}
    \end{multicols}
\end{algorithm}

\begin{definition}[Config Ordering]
  \label{def:config-ordering}
  For configurations $C_i$ and $C_j$, we define the following
  \begin{align*}
    &C_i < C_j \triangleq (C_i.t < C_j.t) \vee (C_i.t = C_j.t \wedge C_i.v < C_j.v)\\
    &C_i > C_j \triangleq  C_j < C_i\\
    &C_i \leq C_j \triangleq C_i < C_j \vee ((C_i.v,C_i.t) = (C_j.v, C_j.t))\\
    &C_i \geq C_j \triangleq C_j \leq C_i
  \end{align*}
\end{definition}
\begin{definition}[Deactivated Config]
  \label{def:deactivated-config}
  A configuration $C$ is deactivated if, for all $Q \in Quorums(C.m)$, there exists some server $n \in Q$ such that $C_n > C$, where $C_n$ is the configuration of server $n$.
  \begin{align*}
    Deactivated(C) \triangleq \forall Q \in Quorums(C.m) : \E n \in Q : C_{(n)} > C
  \end{align*}
\end{definition}
\begin{definition}[Active Config]
  A configuration $C$ is active if it is not deactivated.
  \begin{align*}
      Active(C) \triangleq \neg Deactivated(C)
  \end{align*}
\end{definition}
\begin{definition}[Active Config Set]
  The active config set is the set of servers with a configuration that is active.
  \begin{align*}
      ActiveConfigSet \triangleq \{s \in Server : Active(C_{(s)})\}
  \end{align*}
\end{definition}

% \subsection{Basic Properties of Configurations}
% \label{appendix:basic-config-props}

% Customized command for referring to lemmas in proofs so I can overrride it if needed.
\newcommand{\proofref}[1]{\ref{#1}}
% Optionally highlight referenced lemmas/theorems in bright color.
% \renewcommand{\proofref}[1]{\textcolor{orange}{\textbf{\ref{#1}}}}

\subsection{Election Safety}
\label{appendix:elecsafety-proof}

In this section we present the proof of Lemma \ref{lemma:election-safety} along with auxiliary invariants required for the proof. We prove it inductively, by assuming that Lemmas \ref{lemma:primary-term-equals-cfg-term}, \ref{lemma:config-version-term-unique}, \ref{lemma:primary-in-newest-config-in-term}, \ref{lemma:active-cfgs-overlap}, \ref{lemma:active-cfgs-safe-at-terms}, and \ref{lemma:election-safety} of this section act as strengthening assumptions for the inductive hypothesis needed to prove Lemma \ref{lemma:election-safety}. That is, for the proof of each lemma $L$ in this group, we show that $L$ holds in the initial protocol states, and then show that, if all lemmas of this group hold in the current state, then $L$ holds in the next state, for any possible protocol transition. Lemmas \ref{lemma:deactivated-cannot-reconfig-or-elect}, \ref{lemma:configs-monotonic}, and \ref{lemma:config-deactivations-stable} are some additional, helpful facts that we establish first, and are useful for proving the lemmas in this and later sections.

\begin{lemma}[Deactivated configs cannot reconfig or elect primary]
  \label{lemma:deactivated-cannot-reconfig-or-elect}
  If server $i$ is in a deactivated configuration, $C_i$, then it cannot execute a $Reconfig(i)$ or $BecomeLeader(i)$ action.
\end{lemma}
\begin{proof}
   We must consider \textit{Reconfig} and \textit{BecomeLeader} actions.
  \begin{itemize}
      \item $Reconfig(i)$ requires the \emph{Config Quorum Check} (Algorithm \ref{alg:mrr-pseudocode-full}, Line \ref{mrr-full:line-reconfig-precond-quorum-conds}) to be satisfied for $C_i$, the current configuration of primary server $i$. This requires that, for some quorum $Q \in Quorums(C_i.m)$, all servers in $Q$ are in configuration $C_i$. If $C_i$ is deactivated, though, all quorums of $C_i$ contain some server in configuration $> C_i$, violating this precondition.
      
      \item $BecomeLeader(i)$ elects a primary server $i$ in configuration $C_i$. It requires that a quorum of servers in $C_i$ have configurations $\leq C_i$ (Algorithm \ref{alg:mrr-pseudocode-full}, Line \ref{mrr-full:line-become-leader-config-precond}). If $C_i$ is deactivated, though, all quorums of $C_i$ contain some server in configuration $> C_i$, violating this precondition.
  \end{itemize}
\end{proof}

% Note that Lemmas \ref{lemma:configs-monotonic} and \ref{lemma:config-deactivations-stable} are statements about a pair of states $(s,s')$, rather than a single state. So, the proof of such a lemma $L$ shows that, if all lemmas of this section hold for a state $s$, then $L$ holds for any state transition $(s,s')$. 
\begin{lemma}[Configs increase monotonically]
  \label{lemma:configs-monotonic}
  If Lemma \ref{lemma:active-cfgs-safe-at-terms} holds in the current state, then for all $s \in Server$, if $C_s$ is the configuration of server $s$ in the current state and $C_s'$ is the configuration of $s$ after any state transition, then $C_s' \geq C_s$.
\end{lemma}
\begin{proof}
   We must consider any actions that modify server configurations: \textit{Reconfig}, \textit{BecomeLeader}, \textit{SendConfig}.
  \begin{itemize}
      \item $Reconfig(i)$ updates the configuration on a primary server from $C_i$ to $C_i'$, where $(C_i.m, C_i.t) = (C_i'.m, C_i'.t)$ and $C_i'.v > C_i.v$. So, monotonicity is upheld, by the definition of configuration ordering (Definition \proofref{def:config-ordering}).
      
      \item $BecomeLeader(i)$ elects a primary server $i$ in $term'[i]$ and updates the configuration of $i$ from $C_i$ to $C_i'$, where $(C_i.m, C_i.v) = (C_i'.m, C_i'.v)$ and $C_i'.t = term'[i]$. So, it is sufficient to show that $term'[i] \geq C_i.t$. Assume this were not the case i.e. that $C_i.t > term'[i]$. Since $C_i$ must be active in order for the $BecomeLeader(i)$ action to occur, by Lemma \ref{lemma:deactivated-cannot-reconfig-or-elect}, this would imply that all quorums of $C_i$ contain some server in term $\geq C_i.t$, by Lemma \proofref{lemma:active-cfgs-safe-at-terms}. This would prevent the $BecomeLeader(i)$ action from electing a primary in $term'[i]$, though, since $C_i.t > term'[i]$. So, it must be that $term'[i] \geq C_i.t$, implying that $C_i'.t \geq C_i.t$.
      
      \item $SendConfig(i,j)$ updates a configuration on server $j$ to the configuration of server $i$. It follows directly from the precondition of $SendConfig$ (Algorithm \ref{alg:mrr-pseudocode-full}, Line \ref{mrr-full:line-send-config-precond-gt}) that $C_{(i)} > C_{(j)}$, so monotonicity is upheld, by the definition of configuration ordering (Definition \proofref{def:config-ordering}).
  \end{itemize}
\end{proof}

\begin{lemma}[Config deactivation stability]
  \label{lemma:config-deactivations-stable}
  If Lemma \ref{lemma:configs-monotonic} holds in the current state, then if a configuration $C$ is deactivated in the current state, $C$ cannot be active in the next state.
\end{lemma}
\begin{proof}
  A configuration $C$ is deactivated if, for all $Q \in Quorums(C.m)$, there exists some server $n\in Q$ such that $C_n > C$, where $C_n$ is the configuration of server $n$. Since configurations increase monotonically on servers (Lemma \proofref{lemma:configs-monotonic}), if $C_n > C$ holds currently for some server $n$, it must hold in the next state.
\end{proof}

\begin{lemma}[Primary term equals config term]
  \label{lemma:primary-term-equals-cfg-term}
  For all $i \in Server$, if $i$ is currently primary in $term[i]$ in configuration $C$, then $C.t = term[i]$.
\end{lemma}
\begin{proof}
  In all initial states, no server is primary, so it holds. The only actions that could falsify the lemma are those that change the term of a server's local configuration, its current term, or its primary status: \textit{SendConfig}, \textit{UpdateTerms}, \textit{BecomeLeader}.
  \begin{itemize}
      \item The \textit{SendConfig} action can only update the configuration of a server that is not currently primary, so such a transition could not falsify this lemma in the next state.
      
      \item \textit{UpdateTerms} cannot change the term of a primary server, so it upholds the lemma.
      
      \item $BecomeLeader(i)$ elects server $i$ as primary in $term'[i]$, and it sets $configTerm[i] \leftarrow term'[i]$, so it upholds the lemma.
  \end{itemize}
\end{proof}

\begin{lemma}[Config version and term unique]
  \label{lemma:config-version-term-unique}
  For all servers $i$ and $j$, in configurations $C_i$ and $C_j$, respectively, if $(C_i.v,C_i.t) = (C_j.v,C_j.t)$ then $C_i.m = C_j.m$.
  \begin{align*}
    \forall i,j &\in Server : \\
    &((configVersion[i],configTerm[i]) = (configVersion[j],configTerm[j])) \Rightarrow \\
    &(config[i] = config[j])
  \end{align*}
\end{lemma}
\begin{proof}
  In all initial states, every server configuration is identical, so the lemma holds. The only actions that could falsify Lemma \proofref{lemma:config-version-term-unique} in the next state are those that modify configurations on a server: \textit{BecomeLeader}, \textit{SendConfig}, and \textit{Reconfig}. 
  \begin{itemize}
      \item A $BecomeLeader(i)$ action elects a server $i$ as primary in $term'[i]$ and updates its configuration from $C_i$ to $C_i'$, where $(C_i'.m,C_i'.v) = (C_i.m,C_i.v)$ and $C_i'.t > C_i.t$ (by Lemma \ref{lemma:configs-monotonic}). Since this action only modifies the configuration of server $i$, the only way Lemma \proofref{lemma:config-version-term-unique} could be falsified in the next state is if, in the current state, there was a server $j \neq i$ in configuration $C_j$ such that $(C_j.v,C_j.t) = (C_i'.v,C_i'.t)$ and $C_j.m \neq C_i'.m$. If $C_j.t = C_i'.t$, though, this implies, by Lemma \proofref{lemma:active-cfgs-safe-at-terms}, that all quorums of active configurations in the current state must contain some server in term $\geq C_i'.t$. If this were the case, though, the $BecomeLeader(i)$ action could not have occurred to elect $i$ as primary in term $C_i'.t$, since $C_i$ must be active, by Lemma \ref{lemma:deactivated-cannot-reconfig-or-elect}, and the voting precondition on terms (Algorithm \ref{alg:mrr-pseudocode-full}, Line \ref{mrr-full:line-become-leader-pre-newer-term}) would have prevented it. 
      
      \item $SendConfig(i,j)$ updates the configuration on server $j$ to that of server $i$, and doesn't modify the state of any other server. So, the set of unique configurations in the system can only be reduced or left the same by this action. So, if Lemma \proofref{lemma:config-version-term-unique} held currently, it must hold in the next state.
      
      %
      % Induction counterexample to 'ConfigVersionAndTermUnique' if 
      % 'PrimaryInTermContainsNewestConfigOfTerm' is not satisfied in the current
      % state.
      %
      % State 1: <Initial predicate>
      %   n1 :> "t0 S << >> {} (1,1)" @@
      %   n2 :> "t1 P << >> {n2} (0,1)" )
      %      
      % State 2: <ReconfigAction line 148, col 5 to line 152, col 24 of module MongoRaftReconfig>
      %   n1 :> "t0 S << >> {} (1,1)" @@
      %   n2 :> "t1 P << >> {n2} (1,1)" )
      %
      \item $Reconfig(i)$ updates the configuration on a primary server $i$ from $C_i$ to $C_i'$, where $C_i'.t = term[i]$ and $C_i'.v > C_i.v$. By Lemma \ref{lemma:election-safety}, we know that there is a unique primary per term in the current state, and, since $Reconfig(i)$ doesn't modify the primary status or term of any server, there must still be a unique primary per term in the next state. So, server $i$ is the only primary in $term[i]$, and by Lemma \ref{lemma:primary-in-newest-config-in-term}, $i$ contains the newest configuration in $term[i]$. So, the version of $i$'s new configuration, $C_i'.v$, will be greater than all other configurations in $term[i]$, implying that the configuration $C_i'$ will be unique among all existing configurations, upholding the lemma.
  \end{itemize}
\end{proof}

\begin{lemma}[Primary contains newest config of term]
  \label{lemma:primary-in-newest-config-in-term}
  For all $i \in Server$, if $i$ is currently a primary in $term[i]$, then it contains the newest configuration in $term[i]$.
  \begin{align*}
    \forall i,j &\in Server : \\
    &(state[i] = Primary \wedge configTerm[j] = term[i]) \Rightarrow\\ 
    &(configVersion[j] \leq configVersion[i]) 
  \end{align*}
\end{lemma}
\begin{proof}
  The only actions that could falsify this lemma are those that affect the primary status of a server or modify configurations or terms on a server: \textit{BecomeLeader}, \textit{Reconfig}, \textit{SendConfig}, and \textit{UpdateTerms}.
  \begin{itemize}
      \item $BecomeLeader(i)$ elects a primary server $i$ in $term'[i]$. It updates the configuration of server $i$ from configuration $C_i$ to $C_i'$, where $C_i'.t = term'[i]$ and $C_i.v = C_i'.v$. Since this action doesn't modify the state of any other server, in order for Lemma \proofref{lemma:primary-in-newest-config-in-term} to be falsified in the next state, it must be that, in the current state, there exists some server $j \neq i$, with configuration $C_j$, such that $C_j.t = term'[i]$ and $C_j.v > C_i'.v$. That is, server $j$ contains a configuration in the term of primary $i$ after its election but $j$'s configuration is newer than $i$'s. If $C_j.t = term'[i]$ in the current state, though, this implies, by the assumption of Lemma \proofref{lemma:active-cfgs-safe-at-terms}, that the quorums of all active configurations contain some server in term $\geq C_j.t$. If this were the case, then the $BecomeLeader(i)$ action could not have occurred to elect $i$ as primary in $term'[i]$, since $C_i$ must be active in the current state, by Lemma \ref{lemma:deactivated-cannot-reconfig-or-elect}, and the voting precondition on terms (Algorithm \ref{alg:mrr-pseudocode-full}, Line \ref{mrr-full:line-become-leader-pre-newer-term}) would have prevented it.
      
      \item $Reconfig(i)$ updates the configuration of a primary server $i$ from $C_i$ to $C_i'$, where $C_i'.v > C_i.v$ and $C_i'.t = C_i.t$. By Lemma \proofref{lemma:election-safety}, we know that there is a unique primary in a given term, and since $Reconfig(i)$ doesn't modify the terms or primary status of any server, this will continue to hold in the next state. And, by assumption of Lemma \proofref{lemma:primary-in-newest-config-in-term} in the current state, a primary has the newest configuration in its term. So, the new configuration created by $i$ must be the newest configuration in $term[i]$, and it will be contained on server $i$, the unique primary of $term[i]$.
      
      \item $SendConfig(i,j)$ updates the configuration on a secondary server $j$ to that of server $i$, and doesn't modify the state of any other server. So, no primary configuration is modified and the set of unique configurations in the system can only be reduced or left the same by this action, so if Lemma \proofref{lemma:primary-in-newest-config-in-term} held currently, it must hold in the next state.
      
      \item \textit{UpdateTerms} only updates server terms, and if it updates the term of a server $s$ it sets $state[s] \leftarrow Secondary$, so it could not falsify the lemma.
  \end{itemize}
\end{proof}

%
% Induction counterexample for 'ActiveConfigsOverlap' 
% if 'ConfigVersionAndTermUnique' doesn't hold in the current state:
%
% State 1: <Initial predicate>
%   n1 :> "t0 P << >> {n1, n2} (1,0)" @@
%   n2 :> "t0 S << >> {n2} (1,0)" )
%
% State 2: <ReconfigAction line 148, col 5 to line 152, col 24 of module MongoRaftReconfig>
%   n1 :> "t0 P << >> {n1} (2,0)" @@
%   n2 :> "t0 S << >> {n2} (1,0)" )
%

\begin{lemma}[Active configs overlap]
  \label{lemma:active-cfgs-overlap}
  All quorums of any two active configurations overlap.
  \begin{align*}
      \forall s,t \in ActiveConfigSet : QuorumsOverlap(config[s], config[t])
  \end{align*}
\end{lemma}
\begin{proof}
  In all initial states, there is only a single, unique configuration, so the lemma holds. If Lemma \ref{lemma:active-cfgs-overlap} holds in the current state, then the only actions that could possibly falsify it in the next state are those that affect server configurations: \textit{BecomeLeader}, \textit{SendConfig}, and \textit{Reconfig}. Let $Active$ refer to the set of active configurations that exist in the current state, and $Active'$ to the set of active configurations that exist in the next state.

  \begin{itemize}
      \item $BecomeLeader(i)$ updates the configuration of a server $i$ from $C_i$ to $C_i'$, where $(C_i'.m,C_i'.v) = (C_i.m,C_i.v)$.
      % and $C_i'.t > C_i.t$, by Lemma \proofref{lemma:configs-monotonic}. 
      % $(m,v,t)$ to $(m,v,t')$, where $t'>t$ is the term of the primary $i$ that got elected. 
      This does not change the member set of any existing configuration, and, due to Lemma \ref{lemma:config-deactivations-stable}, cannot create any new active configuration other than $C_i'$, so if the quorums of all configurations in \textit{Active} overlapped, all of those in $Active'$ should still overlap.
      
      \item The $SendConfig(i,j)$ action updates the configuration of server $j$ to $C_i$, the configuration of server $i$. This action does not create any new configurations, so the only way it could falsify Lemma \proofref{lemma:active-cfgs-overlap} in the next state is if it made $C_i$ active in the next state, and the quorums of $C_i$ did not overlap with some other, active configuration. By Lemma \proofref{lemma:config-deactivations-stable}, though, $C_i$ cannot become active in the next state if it was not already active, so the lemma must hold.
      
      \item $Reconfig(i)$ updates the configuration of a primary server $i$ from $C_i$ to $C_i'$. To falsify the lemma in the next state, it must be that $C_i' \in Active'$ and there exists a server $j$, in configuration $C_j$, such that $C_j \in Active'$ and $\neg QuorumsOverlap(C_i'.m, C_j.m)$. We know, by the precondition enforced by the $Reconfig(i)$ action (Algorithm \ref{alg:mrr-pseudocode-full}, Line \ref{mrr-full:line-reconfig-precond-quorums-overlap}), that $QuorumsOverlap(C_i'.m, C_i.m)$, so it must be that $C_j.m \neq C_i.m$, implying, due to Lemma \ref{lemma:config-version-term-unique}, that $(C_j.v, C_j.t) \neq (C_i.v, C_i.t)$.
      In addition, we know that $C_i$ must have been active in order for the $Reconfig(i)$ to occur (Lemma \ref{lemma:deactivated-cannot-reconfig-or-elect}). So, there are now two cases to consider:
    \begin{itemize}
      \item $C_j > C_i$. \\
      We know that $C_i.t = C_i'.t$, by the definition of the $Reconfig$ action. It must also be the case that $C_j.t \neq C_i.t$. Otherwise, it would imply that $C_j.t = C_i.t \wedge C_j.v > C_i.v$, which cannot hold since there is a unique primary per term (Lemma \ref{lemma:election-safety}) and a primary contains the newest configuration of its term (Lemma \ref{lemma:primary-in-newest-config-in-term}). So, it must be that $C_j.t > C_i.t$. If $C_j$ exists, though, by Lemma \proofref{lemma:active-cfgs-safe-at-terms}, all quorums of $C_i$ must contain some server in term $\geq C_j.t$, since we know that $C_i$ is active. But, since $C_j.t > C_i.t$, this would imply the \emph{Term Quorum Check} precondition, $Q2(i)$ (Algorithm \ref{alg:mrr-pseudocode-full}, Line \ref{mrr-full:line-reconfig-precond-quorum-conds}) could not have been satisfied in the current state, preventing the $Reconfig(i)$ action from occurring.

      \item  $C_j < C_i$ \\ 
      By the assumption of Lemma \ref{lemma:active-cfgs-overlap} in the current state, we know that all quorums of $C_j$ and $C_i$ overlap, since both configurations are active. In order for the $Reconfig(i)$ action to occur, though, there must have been some quorum $Q \in Quorums(C_i.m)$ such that, for all servers $n \in Q$, in configuration $C_n$, $(C_n.v, C_n.t) = (C_i.v, C_i.t)$. This is ensured by $Q1(i)$, the \emph{Config Quorum Check} precondition (Algorithm \ref{alg:mrr-pseudocode-full}, Line \ref{mrr-full:line-reconfig-precond-quorum-conds}). If $C_i > C_j$, though, and $QuorumsOverlap(C_i.m, C_j.m)$, this would imply that all quorums of $C_j$ contain some server $m \in Q$, in configuration $C_m$. Since we know that $(C_m.v, C_m.t) = (C_i.v, C_i.t)$ and $C_i > C_j$, this implies that $C_j \notin Active$, which additionally implies that $C_j \notin Active'$, by Lemma \ref{lemma:config-deactivations-stable}, contradicting our assumption that $C_j \in Active'$.
    \end{itemize}
\end{itemize}

% proof-deps,lemma:active-cfgs-overlap,lemma:active-cfgs-safe-at-terms,lemma:config-deactivations-stable,lemma:election-safety

\end{proof}

\begin{lemma}[Active configs safe from past terms]
  \label{lemma:active-cfgs-safe-at-terms}
  For any existing configurations $C_a$ and $C$, where $C_a$ is active, 
  all quorums of $C_a$ contain some server in term $\geq C.t$.
  \begin{align*}
      &\forall s \in Server : \\
      &\forall t \in ActiveConfigSet :\\
      &\forall Q \in Quorums(config[t]) : \exists n \in Q : term[n] \geq configTerm[s]
  \end{align*}
  % \mytodo{Add precise definition.}
\end{lemma}
\begin{proof}
  In all initial states, all servers have the same, unique configuration, and all servers have the same terms, so the lemma holds. We must then consider actions that change configurations or terms on servers: \textit{BecomeLeader}, \textit{Reconfig}, \textit{SendConfig}, and \textit{UpdateTerms}. Again, in the below arguments we refer to $Active$ as the set of active configurations in the current state, and $Active'$ as the set of active configurations in the next state. 

  \begin{itemize}
      \item $BecomeLeader(i)$ updates the configuration on server $i$ from $C_i$ to $C_i'$, where $(C_i'.m, C_i'.v) = (C_i.m, C_i.v)$ and $C_i'.t > C_i.t$ (by Lemma \proofref{lemma:configs-monotonic}). We also know, by Lemma \ref{lemma:deactivated-cannot-reconfig-or-elect}, that $C_i$ must be active in order for the $BecomeLeader(i)$ action to occur. Lemma \proofref{lemma:active-cfgs-safe-at-terms} could be falsified in two cases, which we examine below:
      \begin{itemize}
          \item If $C_i'$ is active, then we must show that, for any other server $j$, with configuration $C_j$, all quorums of $C_i'$ contain some server in term $\geq C_j.t$. By the assumption of Lemma \proofref{lemma:active-cfgs-safe-at-terms} in the current state, we know that all quorums of configuration $C_i$ contain some server in term $\geq C_j.t$. Since $C_i'.m=C_i.m$, this should also hold for $C_i'$.  
          
          \item If $C_j$ is the configuration of some server $j$ such that $C_j \in Active'$, we must show that all quorums of $C_j$ intersect with some server in term $\geq C_i'.t$. By Lemma \proofref{lemma:config-deactivations-stable}, we know that if $C_j \in Active'$ then $C_j \in Active$. So, by the assumption of Lemma \proofref{lemma:active-cfgs-safe-at-terms} in the current state, we know that all quorums of $C_j$ intersect with some server in term $\geq C_i.t$. After the $BecomeLeader(i)$ action occurs, due to its postcondition (Algorithm \ref{alg:mrr-pseudocode-full}, Line \ref{mrr-full:line-become-leader-post-term-update}), there must be some quorum $Q \in Quorums(C_i.m)$  such that $term'[n] = C_i'.t$ for all $n \in Q$, since $C_i.m = C_i'.m$. Since $C_j$ and $C_i$ are both active in the current state, it must be that $QuorumsOverlap(C_i.m, C_j.m)$, by Lemma \ref{lemma:active-cfgs-overlap}. So, all quorums of $C_j$ must contain some server $m \in Q$, where $term[m]\geq C_i'.t$, upholding Lemma \proofref{lemma:active-cfgs-safe-at-terms} in the next state.
          
          % in the current state, all quorums should intersect with some server in term $\geq C_i'.t$ in the next state.
          % \mytodo{Is reasoning about quorum server nodes and overlap correct here?}
      \end{itemize}
      
      \item $Reconfig(i)$ updates the configuration of server $i$ from $C_i$ to $C_i'$, where $C_i.t=C_i'.t$. We know that $C_i$ must have been active in order for the $Reconfig(i)$ to occur, by Lemma \ref{lemma:deactivated-cannot-reconfig-or-elect}. We consider the two cases in which this action could falsify the lemma:
      
      \begin{itemize}
          \item If $C_i'$ is active, then we must show that, for any other server $j$, in configuration $C_j$, all quorums of $C_i'$ overlap with some server in term $\geq C_j.t$. By the assumption of Lemma \proofref{lemma:active-cfgs-safe-at-terms}, we know that, in the current state, all quorums of $C_i$ contain some server in term $\geq C_j.t$, since $C_i \in Active$. In order for the $Reconfig(i)$ action to have occurred, the \emph{Term Quorum Check} precondition, $Q2(i)$ (Algorithm \ref{alg:mrr-pseudocode-full}, Line \ref{mrr-full:line-reconfig-precond-quorum-conds}) must have been satisfied, meaning that there exists some quorum $Q \in Quorums(C_i.m)$ such that,  for all $n \in Q$, $term[n] = C_i.t$. If all quorums of $C_i$ currently contain some server in term $\geq C_j.t$ and $Q2(i)$ was satisfied, though, then this must imply that there is some $m \in Q$ such that
          \begin{align*}
            term[m] = C_i.t \wedge term[m] \geq C_j.t
          \end{align*}
          implying that $C_i.t \geq C_j.t$. So, since $QuorumsOverlap(C_i.m, C_i'.m)$, we know that all quorums of $C_i'$ will contain some server $v$ such that $term[v] \geq C_j.t$, ensuring Lemma \proofref{lemma:active-cfgs-safe-at-terms} is upheld.

          \item If $C_j$ is the configuration of some server $j$ such that $C_j \in Active'$, we must show that all quorums of $C_j$ contain some server in term $\geq C_i'.t$. By assumption of Lemma \proofref{lemma:active-cfgs-safe-at-terms} in the current state, all quorums of $C_j$  contain some server in term $\geq C_i.t$, since $C_i$ is active. Since $C_i.t = C_i'.t$ and $Reconfig(i)$ doesn't modify the terms of any servers, Lemma \proofref{lemma:active-cfgs-safe-at-terms} must be upheld.
      \end{itemize}
      
      \item $SendConfig(i, j)$ updates the configuration of server $j$ to $C_i$, the configuration of server $i$, and does not modify the terms of any servers. The set of unique configurations in the system can only be reduced or left the same by this action. So, it does not create any new active configurations (by Lemma \proofref{lemma:config-deactivations-stable}) and therefore cannot falsify Lemma \proofref{lemma:active-cfgs-safe-at-terms} in the next state.
      
      \item \textit{UpdateTerms} does not modify configurations and can only increase the term of a server, so it must uphold the property.
  \end{itemize}
\end{proof}

%
% The proof of election safety.
%
\begin{proof}[Proof of Lemma \ref{lemma:election-safety}]
  In all initial states, there are no primary servers, so the lemma holds. The only possible actions that could falsify Lemma \proofref{lemma:election-safety} in the next state are $BecomeLeader$ or $UpdateTerms$ actions. 
  
  \begin{itemize}
      \item $BecomeLeader(i)$ elects a primary server $i$ in configuration $C_i$ in $term'[i]$, and does not modify the state of any other server. Assume there exists another server $j \neq i$, in configuration $C_j$, that is also primary in $term'[i]$ in the current state. If server $j$ is currently primary in $term'[i]$, this implies, by Lemma \proofref{lemma:primary-term-equals-cfg-term}, that $C_j.t = term'[i]$. But, in order for $BecomeLeader(i)$ to occur, configuration $C_i$ must be active in the current state, by Lemma \ref{lemma:deactivated-cannot-reconfig-or-elect}. So, if $C_i$ is active, Lemma \proofref{lemma:active-cfgs-safe-at-terms} implies that all quorums of $C_i$ must contain some server in a term $\geq C_j.t = term'[i]$. This would, however, prevent the election of $i$ in $term'[i]$ due to the voting precondition of $BecomeLeader(i)$ (Algorithm \ref{alg:mrr-pseudocode-full}, Line \ref{mrr-full:line-become-leader-pre-newer-term}) that requires some quorum of servers in $C_i$ to have terms $ < term'[i]$.
    
      \item An $UpdateTerms(i,j)$ action only changes the term of server $j$, and sets $state[j] \leftarrow Secondary$, so it cannot falsify the lemma if it held in the current state. 
  \end{itemize}

  % \mytodo{Consider instantiating quorum to be explicit.}
\end{proof}

\subsection{Log Properties}
\label{appendix:log-props}

In this section we establish several auxiliary lemmas related to properties of logs in the system. These lemmas are, for the most part, conceptually unrelated to reconfiguration, but are required for a Raft-based system that replicates logs like \textit{MongoRaftReconfig} and are required for completeness of the proof. Many of the arguments are similar to those in the original Raft dissertation \cite{OngaroDissertation2014}. In Section \ref{appendix:log-matching-proof} we establish the \emph{LogMatching} property (Lemma \ref{lemma:log-matching}) with the help of a few auxiliary lemmas, and in Section \ref{appendix:more-log-lemmas} we establish some additional, higher level lemmas about server logs. All lemmas of the preceding sections (\ref{appendix:elecsafety-proof}) hold in all reachable states of the protocol, so we can utilize them below.

\subsubsection{Log Matching}
\label{appendix:log-matching-proof}

In this section we assume that Lemmas \ref{lemma:log-entry-term-implies-config-term}, \ref{lemma:primary-has-entries-it-created}, and \ref{lemma:log-matching} act as strengthening assumptions for the inductive hypothesis needed to prove Lemma \ref{lemma:log-matching}, the \emph{LogMatching} property. That is, we assume all of these lemmas hold in the current state, and show each is upheld by any protocol transition, as we did in Section \ref{appendix:elecsafety-proof}. Lemma \ref{lemma:get-entries-prefix} is an additional, helpful fact that we establish first.

\begin{lemma}[GetEntries ensures prefix]
  \label{lemma:get-entries-prefix}
  If Lemma \ref{lemma:log-matching} holds in the current state, then after a $GetEntries(i,j)$ action, $log'[i]$ is a prefix of $log[j]$. 
\end{lemma}
\begin{proof}
  $GetEntries(i,j)$ can only occur if $Len(log[j]) > Len(log[i])$ and $log[i][Len(log[i])] = log[j][Len(log[i])]$. If Lemma \proofref{lemma:log-matching} holds currently, then we know that $log[i] = log[j][..Len(log[i])]$. After the $GetEntries(i,j)$ action occurs, $log'[i]=log[i] \circ \langle log[j][Len(log[i])+1] \rangle$.  So, we know that
  \begin{align*}
      log'[i] &= log[i] \circ \langle log[j][Len(log[i])+1] \rangle\\
      &= log[j][..Len(log[i])] \circ \langle log[j][Len(log[i])+1] \rangle\\
      &= log[j][..(Len(log[i])+1)]
  \end{align*}
  showing that $log'[i]$ is a prefix of $log[j]$.
\end{proof}

\begin{lemma}[Log entry in term implies config in term]
  \label{lemma:log-entry-term-implies-config-term}
  If a log entry $E=(ind,t)$ exists in the log of some server, then there exists some server in a configuration $C$ such that $C.t \geq t$.
\end{lemma}
\begin{proof}
  In all initial states, the logs of all servers are empty, so the lemma holds.
  We must consider those actions which modify server logs, terms, or configurations: $ClientRequest$, $GetEntries$, $RollbackEntries$, $UpdateTerms$, $BecomeLeader$, $Reconfig$, and $SendConfig$.

  \begin{itemize}
      \item $ClientRequest(i)$ creates a log entry $(ind,term[i])$ on a primary $i$ in configuration $C_i$. To show that Lemma \ref{lemma:log-entry-term-implies-config-term} holds in the next state we must ensure that there exists some configuration $C$ where $C.t \geq term[i]$. By Lemma \ref{lemma:primary-term-equals-cfg-term}, we know that $C_i.t = term[i]$. Since $C_i.t \geq term[i]$, and the configuration of $i$ is unmodified by this action, Lemma \ref{lemma:log-entry-term-implies-config-term} is upheld in the next state.
      
      \item After a $GetEntries(i,j)$ action occurs, the log of the receiving server, $log'[i]$, is a prefix of the sender's log, $log[j]$, by Lemma \ref{lemma:get-entries-prefix}. So, if Lemma \ref{lemma:log-entry-term-implies-config-term} was satisfied for $log[j]$, it will be satisfied for $log'[i]$, since $log'[i]$ will contain a subset of the entries present in $log[j]$, and no configurations are modified by this action.
      
      \item \textit{RollbackEntries} only deletes log entries, so it cannot falsify the lemma in the next state.  
      
      \item \textit{UpdateTerms} only modifies server terms or primary status, so it cannot falsify the lemma in the next state.

      \item $BecomeLeader(i)$ updates the configuration of a primary server $i$ from $C_i$ to $C_i'$, where $C_i' > C_i$ (by Lemma \ref{lemma:configs-monotonic}), and the action does not modify any server logs. Since it only modifies the configuration of server $i$, and $C_i' > C_i$, Lemma \ref{lemma:log-entry-term-implies-config-term} must hold in the next state if it held currently.
      
      \item $Reconfig(i)$ updates the configuration of server $i$ from $C_i$ to $C_i'$, and does not modify any server terms or logs. Since $C_i'.t=C_i.t$, it must uphold the lemma. 
      
      \item $SendConfig(i,j)$ updates the current configuration of a server $j$ from $C_j$ to $C_j'$, and does not modify any server logs, other configurations, or terms. Since $C_j' > C_j$, by Lemma \ref{lemma:configs-monotonic}, Lemma \ref{lemma:log-entry-term-implies-config-term} must be upheld.
      
  \end{itemize}
\end{proof}

\begin{lemma}[Primary has entries it created]
  \label{lemma:primary-has-entries-it-created}
  For any log entry $E=(ind, t)$ that exists on some server, if a server $s$ is primary in term $t$, then $log[s]$ must contain $E$.
\end{lemma}
\begin{proof}
  In all initial states, the logs of all servers are empty, so the lemma holds. The actions that could possibly falsify this lemma in the next state are those that modify server logs or primary status: \textit{BecomeLeader}, \textit{ClientRequest}, \textit{GetEntries}, and \textit{RollbackEntries}.
  \begin{itemize}
      \item $BecomeLeader(i)$ elects a primary $i$ in $term'[i]$ in configuration $C_i$, and we know that $i$ is the unique primary of $term'[i]$, by Lemma \ref{lemma:election-safety}. We also know that $C_i$ is active, otherwise the election could not have occurred (Lemma \ref{lemma:deactivated-cannot-reconfig-or-elect}). The only way Lemma \ref{lemma:primary-has-entries-it-created} could be falsified in the next state is if, in the current state, there exists some server $j$ such that there is a log entry $E_j=(ind_j,term'[i])$ in $log[j]$ and $E_j$ is not contained in $log[i]$. If $E_j$ exists, though, by Lemma \ref{lemma:log-entry-term-implies-config-term} this would imply that, in the current state, there exists some configuration in term $\geq term'[i]$. Then, by Lemma \ref{lemma:active-cfgs-safe-at-terms}, this would imply that all quorums of active configurations in the current state contain some server in term $\geq term'[i]$. This means that all quorums of $C_i$ would contain some server in term $\geq term'[i]$, preventing the $BecomeLeader(i)$ action from occurring, due to the voting precondition on terms (Algorithm \ref{alg:mrr-pseudocode-full}, Line \ref{mrr-full:line-become-leader-pre-newer-term}).
      
      \item $ClientRequest(i)$ appends a new log entry $E=(ind,term[i])$ to the log of a primary server $i$. Since there is a unique primary per term (Lemma \ref{lemma:election-safety}), we are assured that $E$ is present in the log of server $i$, the only primary in $term[i]$, since $i$ is the server that created the entry.
      
      \item $GetEntries(i,j)$ sends a new log entry from server $j$ to a secondary server $i$. We know that $log'[i]$ is a prefix of $log[j]$, by Lemma \ref{lemma:get-entries-prefix}. So, the entries contained in $log'[i]$ are a subset of those in $log[j]$. Thus, if Lemma \ref{lemma:primary-has-entries-it-created} held currently for all entries in $log[j]$, it will hold in the next state for all entries in $log'[i]$, since the logs of no primary servers are modified.

      \item \textit{RollbackEntries} only modifies log entries on a secondary server, so it cannot falsify Lemma \ref{lemma:primary-has-entries-it-created} in the next state. 
  \end{itemize}

  % proof-deps,lemma:primary-has-entries-it-created,lemma:election-safety,lemma:log-entry-term-implies-config-term,lemma:active-cfgs-safe-at-terms,lemma:get-entries-prefix

\end{proof}

% Note: I don't think LogMatching directly depends on LogEntryInTermImpliesConfigInTerm. I believe it should depend on PrimaryHasEntriesItCreated, which depends on LogEntryInTermImpliesConfigInTerm.

%
% 1-step counterexample to induction for 
% LogMatching that arises if PrimaryHasEntriesItCreated 
% doesn't hold in the current state:
%
% <State 1>
% P1 : [0, 2]
% P2 : [2]
% 
% <State 2 (ClientRequestAction)>
% P1 : [0, 2]
% P2 : [2, 2]
%

\begin{lemma}[Log Matching]
  \label{lemma:log-matching}
  An $(index, term)$ pair uniquely identifies a log prefix.
  \begin{align*}
    \forall& s,t \in Server :\\
    \forall& ind \in (1..Len(log[s]) \cap 1..Len(log[t])) : \\
    &(log[s][ind] = log[t][ind]) \Rightarrow (log[s][..ind] = log[t][..ind])
  \end{align*}
\end{lemma}
\begin{proof}
  In all initial states, the logs of all servers are empty, so the lemma holds. If Lemma \ref{lemma:log-matching} holds in the current state, the only possible actions that could falsify it in the next state are those that affect the state of server logs: \textit{ClientRequest}, \textit{GetEntries}, and \textit{RollbackEntries}. 
  
  \begin{itemize}
      \item $ClientRequest(i)$ appends a single entry to the log of a primary server $s$ in $term[i]$. Let $ind_i = Len(log[i])$. So, we know that 
      \begin{align}
        \label{eq:log-matching-a}
        log'[ind_i+1] = term[i] 
      \end{align}
      The only way for this action to violate Lemma \ref{lemma:log-matching} in the next state is if there exists some server $j \neq i$ such that both of the following hold:
      \begin{align}
        \label{eq:log-matching-b}
        &log'[i][ind_i+1] = log[j][ind_i+1]\\ 
        \label{eq:log-matching-c}
        &log'[i][..(ind_i+1)] \neq log[j][..(ind_i+1)]
      \end{align}
      That is, $log[j]$ contains an entry in $term[i]$ at index $ind_i+1$, but it has a prefix that differs from $log'[i]$. This cannot be possible, though, since, from statements \ref{eq:log-matching-a} and \ref{eq:log-matching-b} above, we know that 
      $log[j][ind_i+1] = term[i]$. By Lemma \proofref{lemma:primary-has-entries-it-created}, this implies that the entry $(ind_i+1, term[i])$ must be contained in $log[i]$. So, it must be that $Len(log[i]) \geq ind_i+1$, contradicting our assumption that $Len(log[i]) = ind_i$.
      
      \item $GetEntries(i,j)$ sends a single log entry from server $i$ to server $j$, 
      and, by Lemma \proofref{lemma:get-entries-prefix}, we know that $log'[i]$ is a prefix of $log[j]$ after the action occurs. So, if Lemma \proofref{lemma:log-matching} held in the current state, it will hold in the next state.

      \item \textit{RollbackEntries} truncates a single entry from the end of a server's log, so it cannot not violate Lemma \proofref{lemma:log-matching}. 
  \end{itemize}

  % proof-deps,lemma:log-matching,lemma:primary-has-entries-it-created

  % \mytodo{See if, in this argument, I can understand precsiely how this lemma relies on the two helper lemmas above. Remember, we are not making many assumptions about logs at this point e.g. no monotonicity, etc.}
\end{proof}

\subsubsection{Additional Log Lemmas}
\label{appendix:more-log-lemmas}
% We again assume that all of the lemmas in this section (Lemmas \ref{lemma:primary-term-gte-log-term}, \ref{lemma:log-terms-monotonic}, and \ref{lemma:uniform-log-entries}) act as strengthening assumptions for the inductive hypothesis needed to prove Lemma \ref{lemma:uniform-log-entries}. That is, we assume they all hold in the current state, and show each is upheld by any protocol transition. 

In this section we prove some additional log lemmas. Each is proved inductively. We assume that all lemmas established in the preceding sections (\ref{appendix:elecsafety-proof}, \ref{appendix:log-matching-proof}) hold in all reachable states of the protocol, so we can utilize them below.

\begin{lemma}[Primary term at least as large as log term]
  \label{lemma:primary-term-gte-log-term}
  For any server that is currently primary, its current term must be $\geq$ the largest term of any entry in its log.
  \begin{align*}
      \A s \in Server : (state[s] = Primary) \Rightarrow \forall ind \in 1..Len(log[s]) : term[s] \geq log[s][ind]
  \end{align*}
\end{lemma}
\begin{proof}
  In all initial states, the logs of all servers are empty, so the lemma holds. We must consider actions that modify the primary status of a server, terms of a server, or its logs: \textit{BecomeLeader}, \textit{ClientRequest}, \textit{GetEntries}, \textit{RollbackEntries}, \textit{UpdateTerms}.
  \begin{itemize}

      %
      % Induction counterexample for 'PrimaryTermAtLeastAsLargeAsLogTerms' if 
      % 'ActiveConfigsSafeAtTerms' doesn't hold in the current state:
      %
      % State 1: <Initial predicate>
      %   n1 :> "t0 S <<1>> {} (0,0)" @@
      %   n2 :> "t0 S <<2, 2>> {n2} (1,2)" )
      %
      % State 2: <BecomeLeaderAction line 159, col 26 to line 161, col 33 of module MongoRaftReconfig>
      %   n1 :> "t0 S <<1>> {} (0,0)" @@
      %   n2 :> "t1 P <<2, 2>> {n2} (1,1)" )
      %

      \item $BecomeLeader(i)$ elects a new primary server $i$ in $term'[i]$. It also updates the configuration of $i$ from $C_i$ to $C_i'$, where $(C_i.m, C_i.v) = (C_i'.m, C_i'.v)$ and $C_i'.t > C_i.t$ (by Lemma \ref{lemma:configs-monotonic}). To uphold Lemma \proofref{lemma:primary-term-gte-log-term}, we must be sure that $term'[i]$ is $\geq$ the largest term of any entry in $log[i]$. Assume there was an index $ind$ such that $log[i][ind] > term'[i]$. By Lemma \proofref{lemma:log-entry-term-implies-config-term}, this implies that there exists some server $j$, in configuration $C_j$, such that $C_j.t = log[i][ind]$. By Lemma \proofref{lemma:active-cfgs-safe-at-terms}, this implies that all quorums of $C_i$, which must have been active in the current state for the $BecomeLeader(i)$ to occur (Lemma \ref{lemma:deactivated-cannot-reconfig-or-elect}), contain some server in term $\geq log[i][ind]$. Since $log[i][ind] > term'[i]$, this would have prevented the $BecomeLeader(i)$ action from occurring, due to its voting precondition on terms (Algorithm \ref{alg:mrr-pseudocode-full}, Line \ref{mrr-full:line-become-leader-pre-newer-term}).
      
      \item $ClientRequest(i)$ creates a new entry $(ind, term[i])$ on a primary in $term[i]$, and does not modify the state of any other server. If Lemma \ref{lemma:primary-term-gte-log-term} holds currently, then, since the term of the new log entry is $ term[i]$, it will continue to hold in the next state, since $term[i] \leq term[i]$.

      \item $GetEntries$ only modifies the state of server logs on a secondary server, so could not falsify the lemma.
      
      \item $RollbackEntries$ only removes log entries from a secondary server, so it could not falsify the lemma if it holds in the current state.
      
      \item $UpdateTerms$ only increases the term of a server and does not modify any server logs, so it could not violate the lemma if it holds currently. 
  \end{itemize}

  % proof-deps,lemma:primary-term-gte-log-term,lemma:log-entry-term-implies-config-term,lemma:active-cfgs-safe-at-terms
  
\end{proof}

\begin{lemma}[Log entry terms increase monotonically]
  \label{lemma:log-terms-monotonic}
  For all $s \in Server$, the terms of the log entries in $log[s]$ increase monotonically.
  \begin{align*}
    \forall s \in Server : \A ind_i,ind_j \in 1..Len(log[s]) : (ind_i < ind_j) \Rightarrow log[s][ind_i] \leq log[s][ind_j]
  \end{align*}
\end{lemma}
\begin{proof}
  In all initial states, the logs of all servers are empty, so the lemma holds. We only need to consider actions that modify server logs: \textit{ClientRequest}, \textit{GetEntries}, and \textit{RollbackEntries}.
  \begin{itemize}
    \item A $ClientRequest(i)$ appends a new log entry $(ind, term[i])$ on a primary $i$ in term $term[i]$. By Lemma \proofref{lemma:primary-term-gte-log-term} we know that $term[i]$ is $\geq$ the largest term of any entry in $log[i]$. So, $log'[i][ind]$ must be $\geq$ the term of the largest entry in $log[i]$, ensuring monotonicity of log entry terms in $log'[i]$.

    \item $GetEntries(i,j)$ ensures that the log of the receiving server, $i$, log is a prefix of the sender, $j$, by Lemma \proofref{lemma:get-entries-prefix}. So, if the sender's log satisfies Lemma \proofref{lemma:log-terms-monotonic} then the receiver's also will.
      
    \item \textit{RollbackEntries} only deletes log entries, so it must maintain Lemma \proofref{lemma:log-terms-monotonic}.
  \end{itemize}

  % proof-deps,lemma:log-terms-monotonic,lemma:get-entries-prefix,lemma:primary-term-gte-log-term

\end{proof}

%
% Note: this lemma shouldn't actually depend on the two above lemmas in this section.
%
\begin{lemma}[Uniform log entries in term]
  \label{lemma:uniform-log-entries}
  For all $i,j \in Server$, if $log[i]$ contains a log entry $(ind_i,t)$ and $log[j]$ contains an entry $(ind_j,t)$, where $ind_j < ind_i$, it must be that $log[i][ind_j]=t$. 
  \begin{align*}
      \A &i,j \in Server :\\
      \A &ind_i \in 1..Len(log[i]) :\\
      \A &ind_j \in 1..Len(log[j]) : \\
      &(ind_j < ind_i \wedge log[i][ind_i] = log[j][ind_j]) \Rightarrow (log[i][ind_j] = log[i][ind_i])
  \end{align*}
\end{lemma}
\begin{proof}
  In all initial states, the logs of all servers are empty, so the lemma holds. The only actions that could falsify this lemma are those that modify logs of servers: \textit{ClientRequest}, \textit{GetEntries}, \textit{RollbackEntries}. 

  \begin{itemize}
      \item $ClientRequest(i)$ appends a new log entry $(ind_i, term[i])$ on primary server $i$ in $term[i]$. The only way this action could falsify Lemma \ref{lemma:uniform-log-entries} in the next state is in the following two cases:
      
      \begin{itemize}
        \item There is another server $j$ that contains an entry $(ind_j,term[i])$ such that 
        \begin{align*}
          \, & ind_j < ind_i\\
          \wedge \, & log'[i][ind_j] \neq term[i]
        \end{align*}
        By Lemma \proofref{lemma:primary-has-entries-it-created}, we know that a primary server has all log entries that exist in its own term, so if $log[j][ind_j] = term[i]$, it must be that $log'[i][ind_j] = term[i]$, since server $i$ is primary in $term[i]$ in both the current and next state.

        \item There is another server $j$ that contains an entry $(ind_j,term[i])$ such that 
        \begin{align*}
          \, & ind_i < ind_j\\
          \wedge \, &log[j][ind_i] \neq term[i]
        \end{align*}
        By Lemma \proofref{lemma:primary-has-entries-it-created}, we know that a primary server must have all entries in its term, so if $(ind_j,term[i])$ exists, it must be contained in $log'[i]$, which means that $ind_i \geq ind_j$, contradicting our assumption that $ind_i < ind_j$.
      \end{itemize}

      \item For a $GetEntries(i,j)$ action, we know that the receiver's log after the action is a prefix of the sender's log, by Lemma \proofref{lemma:get-entries-prefix}, so this action could not falsify the lemma, since it held currrently for the sender's log.

      \item \textit{RollbackEntries} only deletes log entries, so it could not falsify the lemma in the next state. 
  \end{itemize}

  % proof-deps,lemma:uniform-log-entries,lemma:get-entries-prefix,lemma:primary-has-entries-it-created

\end{proof}

\subsection{Leader Completeness}
\label{appendix:leader-comp-proof}

In this section we present the proof of Theorem \ref{thm:leader-completeness}, which relies on the auxiliary lemmas of this section and those proven in the previous sections. Similar to the proofs in the preceding sections, we assume that all of the lemmas in this section (Lemmas \ref{lemma:newer-logs-contain-committed}, \ref{lemma:active-config-quorums-intersect-committed}, \ref{lemma:newer-configs-disable-commits} and Theorem \ref{thm:leader-completeness})  act as strengthening assumptions for the inductive hypothesis needed to prove Theorem \ref{thm:leader-completeness}. All lemmas established in the preceding sections (\ref{appendix:elecsafety-proof}, \ref{appendix:log-props}) hold in all reachable states of the protocol, so we can utilize them below.

\begin{lemma}[Logs later than committed must have past committed entries]
  \label{lemma:newer-logs-contain-committed}
  If a log contains an entry $(ind,t)$, then it also contains all entries committed in terms $< t$.
  \begin{align*}
      \forall& s \in Server: \\
      \forall& (index, term) \in committed : \\ 
      \forall& ind_s \in 1..Len(log[s]) : \\
        &term < log[s][ind_s] \Rightarrow InLog(index, term, s)
  \end{align*}
\end{lemma}
\begin{proof}
  In all initial states, the logs of all servers are empty, so the lemma holds. We must consider the actions that modify server logs or the set of committed entries: \textit{ClientRequest}, \textit{GetEntries}, \textit{RollbackEntries}, and \textit{CommitEntry}.
  
  \begin{itemize}
      \item $ClientRequest(i)$ appends a new, uncommitted entry $(ind,term[i])$ to the log of primary server $i$ in term $term[i]$. We must show that $log'[i]$ contains all log entries committed in terms $< term[i]$. By the assumption of Theorem \ref{thm:leader-completeness} in the current state, we know that server $i$, which is primary in $term[i]$, contains all log entries committed in terms $< term[i]$. The newly appended entry of $(ind,term[i])$ is uncommitted, so the set of committed entries is not changed by this action and no other server logs are modified. So, if server $i$ contained all entries committed in terms $< term[i]$ in the current state, it will in the next state.

      \item After a $GetEntries(i,j)$ action, due to Lemma \ref{lemma:get-entries-prefix}, $log'[i]$ is a prefix of $log[j]$ and the set of committed entries is unmodified. So, if $log[j]$ satisfied Lemma \ref{lemma:newer-logs-contain-committed} in the current state, then $log'[i]$ will also satisfy it.
      
      \item \label{proof:active-cfgs-overlap:rollback-entries}$RollbackEntries(i,j)$ removes a single log entry, $E$, from the end of a secondary server $i$'s log. In order for this to falsify the lemma in the next state, it would have to be the case that $E$ is a committed log entry. In order for a $RollbackEntries(i,j)$ action to occur, the $CanRollback(i,j)$ predicate must be satisfied for servers $i$ and $j$. This predicate is satisfied if 
      \begin{align*}
        &(LogTerm(i)< LogTerm(j)) \wedge \neg IsPrefix(log[i], log[j])
      \end{align*}
      If $LogTerm(i) < LogTerm(j)$ and $E=(ind, LogTerm(i))$ were committed in the current state, it would imply that $E_i$ is contained in $log[j]$, by assumption of Lemma \ref{lemma:newer-logs-contain-committed}. If this were the case, though, then 
      \begin{align*}
        log[j][Len(log[i])] = log[i][Len(log[i])]
      \end{align*}
      implying, by Lemma \ref{lemma:log-matching}, that $log[i]$ is a prefix of $log[j]$, contradicting the precondition of $RollbackEntries$ requiring that $log[i]$ cannot be a prefix of $log[j]$.
    
      \item $CommitEntry(i)$ commits the latest log entry $E_i=(ind,term[i])$ of a primary server $i$ in $term[i]$ in configuration $C_i$. In order for this action to falsify Lemma \ref{lemma:newer-logs-contain-committed} in the next state, there must exist some server $j$ such that $log[j]$ contains an entry $E_j=(ind_j,t_j)$ where $t_j > term[i]$ and $E_i$ is not contained in $log[j]$. If $E_j$ exists with term $t_j$, though, this implies, by Lemma \ref{lemma:log-entry-term-implies-config-term}, that there exists a configuration $C_k$, where $C_k.t\geq t_j$. By Lemma \ref{lemma:newer-configs-disable-commits}, this implies that all quorums of primary server $i$'s configuration, $C_i$, must contain some server in term $> term[i]$. This would prevent the $CommitEntry(i)$ action from occurring, due to its precondition that requires a quorum of servers in $C_i$ to be in $term[i]$ (Algorithm \ref{alg:mrr-pseudocode-full}, Line \ref{mrr-full:line-commit-entry-precond}). 
  \end{itemize}

  % proof-deps,lemma:newer-logs-contain-committed,lemma:newer-logs-contain-committed,lemma:log-matching,lemma:newer-logs-contain-committed,get-entries-prefix,thm:leader-completeness,lemma:log-entry-term-implies-config-term,lemma:newer-configs-disable-commits

\end{proof}

\begin{lemma}[Active configs overlap with committed entries]
  \label{lemma:active-config-quorums-intersect-committed}
  For any active configuration $C$ and committed entry $E$, all quorums of $C$ contain some server that has entry $E$ in its log.
  \begin{align*}
    &\forall s \in ActiveConfigSet :\\
    &\forall (index, term) \in committed :\\
    &\forall Q \in Quorums(config[s]) : \exists n \in Q : InLog(index, term, n) 
  \end{align*}
\end{lemma}
\begin{proof}
  In all initial states, the set of committed entries is empty, so the lemma holds. The only actions that could possibly falsify this property in the next state are those that delete log entries, modify configurations, or update the set of committed entries: \textit{RollbackEntries}, \textit{Reconfig}, \textit{SendConfig}, \textit{BecomeLeader}, and \textit{CommitEntry}.
    
    \begin{itemize}
        \item $RollbackEntries(i,j)$ truncates one entry from the log of a secondary server. So, in order for this action to falsify Lemma \ref{lemma:active-config-quorums-intersect-committed} in the next state, it must delete a committed log entry on some server. As argued in the $RollbackEntries$ case of the proof of Lemma \ref{lemma:newer-logs-contain-committed}, though, \textit{RollbackEntries} cannot delete a committed log entry.

        \item $Reconfig(i)$ updates the configuration of a primary server $i$ in $term[i]$ from $C_i$ to $C_i'$. We know that no new, active configurations on other servers could be created by this action since no configurations on other servers are changed, and deactivated configurations cannot become active (Lemma \ref{lemma:config-deactivations-stable}). Also, server logs and the set of committed entries are not modified. So, we only need to show that, if $C_i'$ is active, then, for all committed entries $E$, all quorums of $C_i'$ contain some server that has $E$ in its log. Suppose there is some quorum $Q \in Quorums(C_i'.m)$ and committed entry $E_j=(ind_j,t_j)$ such that no server in $Q$ contains entry $E_j$ in its log. By the assumption of Lemma \ref{lemma:active-config-quorums-intersect-committed} in the current state, we know that all quorums of $C_i$ contain some server with $E_j$ in its log, since $C_i$ is active (by Lemma \ref{lemma:deactivated-cannot-reconfig-or-elect}). Furthermore, we know that the \textit{Oplog Commitment} precondition, $P1(i)$, must have been satisfied in the current state in order for the $Reconfig(i)$ action to occur. This implies that, for all entries committed in $term[i]$, some quorum $Q_i \in Quorums(C_i.m)$ contains this log entry and all servers $n \in Q_i$ are in $term[i]$. Now, consider the following two cases:
        \begin{itemize}
            \item $t_j \leq term[i]$\\ 
            We know that all servers in $Q_i$ contain all entries committed in $term[i]$. So, by Lemma \ref{lemma:newer-logs-contain-committed}, we also know that the logs of all servers in $Q_i$ contain all entries committed in terms $\leq term[i]$. So, the log of every server in $Q_i$ must contain $E_j$. Since the quorums of $C_i$ and $C_i'$ overlap, all quorums of $C_i'$ must contain some server that has entry $E_j$ in its log, upholding Lemma \ref{lemma:active-config-quorums-intersect-committed}. 
            
            % Note: in order to execute a Reconfig, you must be in an active configuration (??) Perhaps make this explicit if necessary.
            \item $t_j > term[i]$\\
            By Lemma \ref{lemma:log-entry-term-implies-config-term}, this implies there exists a configuration $C_j$ such that $C_j.t=t_j$. If $C_j$ exists, though, then this implies that the $Reconfig(i)$ could not have occurred, due to Lemma \ref{lemma:active-cfgs-safe-at-terms} and the \textit{Term Quorum Check} precondition, $Q2(i)$ (Algorithm \ref{alg:mrr-pseudocode-full}, Line \ref{mrr-full:line-reconfig-precond-quorum-conds}), which requires a quorum of servers in $C_i$ to be in $term[i]$. Since $C_j.t = t_j$ and $t_j > term[i]$, the $Reconfig(i)$ would have been prevented.
        \end{itemize}
        
        \item $SendConfig(i,j)$ cannot create new  configurations, cannot activate any existing configurations (by Lemma \ref{lemma:config-deactivations-stable}), and does not modify the logs of any servers or the set of committed entries, so it must uphold Lemma \ref{lemma:active-config-quorums-intersect-committed}.
        
        \item $BecomeLeader(i)$ updates the configuration on server $i$ from $C_i$ to $C_i'$, where $(C_i'.m, C_i'.v) = (C_i.m, C_i.v)$. It does not change the member set of any existing configuration and does not modify the set of committed entries or server logs. So, if $C_i'$ is active in the next state, we know that its quorums will overlap with some server containing a committed entry, for all committed entries, since we know this held for $C_i$, and $C_i.m=C_i'.m$.

        \item $CommitEntry(i)$ commits a log entry $E_i$ on a primary server $i$ in configuration $C_i$, and does not modify server logs or configurations. In order for this to falsify the lemma in the next state, there must exist some server $j$, in active configuration $C_j$, and some quorum $Q_j \in Quorums(C_j.m)$ such that no server in $Q_j$ contains entry $E_i$. If a primary $i$ can commit a log entry in $C_i$, this means that there is a quorum $Q_i \in Quorums(C_i.m)$ such that all servers in $Q_i$ contain $E_i$ and are in $term[i]$. This implies that $\neg QuorumsOverlap(C_i.m,C_j.m)$, since otherwise $Q_i \cap Q_j \neq \emptyset$, implying $Q_j$ contains some server with entry $E_i$. So, this must imply that $C_i.m \neq C_j.m$, which, by Lemma \ref{lemma:config-version-term-unique}, implies that $(C_i.v, C_i.t) \neq (C_j.v, C_j.t)$. So, we have the following two cases to consider:
        
        % In addition, since $C_j$ is active and does not overlap with $Q_i$, a quorum of $C_i$, this must imply that $C_i$ is deactivated. Otherwise, their quorums would overlap, by Lemma \ref{lemma:active-cfgs-overlap}. Now, consider the following cases:
        
        \begin{itemize}

            \item $C_i < C_j$ \\
            If $C_i.t = C_j.t$, this would imply that server $j$ had a newer configuration, $C_j$, in term $C_i.t$, which, by Lemma \ref{lemma:primary-in-newest-config-in-term}, could not be possible since primary $i$ contains the newest configuration of its term. So, it must be that $C_i.t < C_j.t$. By Lemma \ref{lemma:newer-configs-disable-commits} this implies that $C_i$ is prevented from executing a $CommitEntry(i)$ action in $term[i]$, since all quorums of $C_i$ must intersect with some server in term $> term[i]$.

            \item $C_i > C_j$ \\
            If primary $i$ is able to commit an entry in $term[i]$, the term of its configuration, $C_i$, must be greater than all other configurations. If another configuration existed in a term $> term[i]$, it would prevent a $CommitEntry(i)$ action, by Lemma \ref{lemma:newer-configs-disable-commits}. This means that configuration $C_i$ is the newest configuration in existence, since there are no configurations that exist in higher terms, and we know that primary $i$ contains the newest configuration of its term, by Lemma \ref{lemma:primary-in-newest-config-in-term}. If, however, there are no configurations newer than $C_i$, this implies that $C_i$ must be active, since the definition of deactivation requires the existence of at least some configuration $> C_i$. If $C_i$ is active, then by Lemma \ref{lemma:active-cfgs-overlap}, this means that all quorums of $C_i$ and $C_j$ overlap, implying that $Q_i \cap Q_j \neq \emptyset$. So, $Q_j$ must contain some server that has entry $E_i$.
        \end{itemize}

    \end{itemize}
\end{proof}

\begin{lemma}[Newer configs disable commits in older terms]
  \label{lemma:newer-configs-disable-commits}
  For all servers $i,j \in Server$, in configurations $C_i$ and $C_j$, respectively, if server $i$ is primary and $C_j.t > term[i]$, then all quorums of $C_i$ must contain some server in term $> term[i]$.
  \begin{align*}
    \forall &s,t \in Server : \\
    &(state[t] = Primary \wedge term[t] < configTerm[s]) \Rightarrow \\
        & \forall Q \in Quorums(config[t]) : \exists n \in Q : term[n] > term[t]
  \end{align*}
\end{lemma}
\begin{proof}
    In all initial states, no servers are primary, so the lemma must hold. We must consider actions that modify configurations, primary status of a server, or terms of servers: \textit{Reconfig}, \textit{SendConfig}, \textit{BecomeLeader}, and \textit{UpdateTerms}.
    
    \begin{itemize}
        \item $Reconfig(i)$ updates the configuration on a primary server $i$ from $C_i$ to $C_i'$.
        % where $C_i'.t > C_i.t$ and $C_i'.m = C_i.m$. 
        For this action to falsify Lemma \ref{lemma:newer-configs-disable-commits} in the next state, there are two cases:
        \begin{itemize}
            \item There exists some server $j$, in configuration $C_j$, such that $C_j.t > term[i]$, and some quorum $Q \in Quorums(C_i'.m)$ that does not contain a server in term $> term[i]$. By assumption of Lemma \ref{lemma:newer-configs-disable-commits} in the current state, though, we know that all quorums of $C_i$ contain some server in term $> term[i]$. In order for the $Reconfig(i)$ to occur, the \textit{Term Quorum Check} precondition, $Q2(i)$, must have been satisfied, which requires that there exists some quorum $Q_i \in Quorums(C_i.m)$ such that $term[n] = term[i]$, for all $n \in Q_i$. This, however, contradicts the assumption that all quorums of $C_i$ contain  some server in term $> term[i]$. So, such a $Reconfig(i)$ action could not have occurred.
            
            \item There exists some primary server $j \neq i$, such that $C_i'.t > term[j]$, and there is some quorum $Q \in Quorums(C_j.m)$ such that $Q$ does not contain a server in term $> term[j]$. Since we know that $C_i'.t = C_i.t$ and $C_i'.t > term[j]$, we know, by assumption of Lemma \ref{lemma:newer-configs-disable-commits} in the current state, that all quorums of $C_j$ contain some server in term $> term[j]$. So, the lemma must continue to hold in the next state.
        \end{itemize}

        % Alternatively, the lemma could also be falsified if there exists another current primary server $j \neq i$, such that $C_i'.t > C_j.t$ and some quorum of $C_j$ does not contain a server in term $> C_j.t$. If $C_i'.t$ exists, tthough, then assumption of Lemma 24 in the current state ensures this will hold
        
        \item $SendConfig(i,j)$ transfers the configuration of server $i$ to a secondary server $j$. This action does not create any new configurations, and it doesn't modify terms or the configuration of a primary server, so if the lemma held for all configurations in the current state, it will hold in the next state.

        \item $BecomeLeader(i, Q)$ elects a server $i$ as primary with voters $Q$ and updates its configuration from $C_i$ to $C_i'$, where $C_i'.t > C_i.t$, by Lemma \ref{lemma:configs-monotonic}. In order to falsify this lemma there are two cases:
        \begin{itemize}
            % Note for future: this case might be able to be combined with some of the Reconfig(i) cases.
            \item There exists another server $j$, in configuration $C_j$, such that $C_j.t > C_i'.t$ and there is some quorum $Q_i \in Quorums(C_i'.m)$ such that $Q_i$ does not contain a server in term $> C_i'.t$. If $C_j$ exists and $C_j.t > C_i'.t$, then Lemma \ref{lemma:active-cfgs-safe-at-terms} implies that all quorums of $C_i$ must contain some server in term $\geq C_j.t$, since $C_i$ must be active in order for a $BecomeLeader(i)$ action to occur (Lemma \ref{lemma:deactivated-cannot-reconfig-or-elect}). Since $C_j.t > C_i'.t > C_i.t$, this would have prevented the $BecomeLeader(i)$ action from occurring, due to its voting precondition on terms.
            
            \item There exists some primary server $j \neq i$, such that $C_i'.t > term[j]$, and there is some quorum $Q_j \in Quorums(C_j.m)$ such that $Q_j$ does not contain a server in term $> term[j]$. If $C_j$, the configuration of server $j$, is deactivated, then this implies the existence of some configuration $> C_j$ in the current state which, by the assumption of Lemma \ref{lemma:newer-configs-disable-commits} in the current state, would ensure that all quorums of $C_j$ contain some server in term $> term[j]$, since $j$ is primary, which ensures Lemma \ref{lemma:newer-configs-disable-commits}. If $C_j$ is active, then both it and $C_i$ must be active, since the $BecomeLeader(i)$ action was able to occur (Lemma \ref{lemma:deactivated-cannot-reconfig-or-elect}).  If both $C_i$ and $C_j$ are active, this implies that $QuorumsOverlap(C_i.m, C_j.m)$. After the $BecomeLeader(i)$ action occurs, due to its postcondition, all servers in $Q$ will have a term of $C_i'.t > term[j]$. Since $C_i'.m = C_i.m$, this means that a quorum of servers in $C_i$ will also have a term of $C_i'.t$. So, all quorums of $C_j$ will contain some server in term $> term[j]$, upholding Lemma \ref{lemma:newer-configs-disable-commits}..
        \end{itemize}
        
        \item An \textit{UpdateTerms} action could not falsify the lemma in the next state since it only updates server terms, which increase monotonically on all servers, and if it updates the term of a server $s$ it sets $state[s] \leftarrow Secondary$.
    \end{itemize}
\end{proof}

%
% Induction counterexample to 'LeaderCompletenessGeneralized' when 'UniformLogEntries'
% does not hold in the current state:
%
% State 1: <Initial predicate>
%   n1 :> "t1 S <<2>> {n1, n2} (1,2)" @@
%   n2 :> "t2 S <<1, 2>> {n1, n2} (1,2)" )
%   committed = {[term |-> 2, entry |-> <<1, 2>>]}
%
% State 2: <BecomeLeaderAction line 159, col 26 to line 161, col 33 of module MongoRaftReconfig>
%   n1 :> "t3 S <<2>> {n1, n2} (1,2)" @@
%   n2 :> "t3 P <<1, 2>> {n1, n2} (1,3)" 
%   committed = {[term |-> 2, entry |-> <<1, 2>>]}
%

%
% Induction counterexample to 'LeaderCompletenessGeneralized' when 'LogsLaterThanCommittedMustHaveCommitted' does
% not hold in the current state.
%
% State 1: <Initial predicate>
%   n1 :> "t1 S <<0, 0>> {n2, n3} (2,2)" @@
%   n2 :> "t0 S <<0, 0>> {n1} (2,0)" @@
%   n3 :> "t2 S <<1, 1>> {n2, n3} (1,1)" )
% /\ committed = {[term |-> 0, entry |-> <<1, 0>>]}
% /\ activeConfigs = {n1, n3}
%
% State 2: <BecomeLeaderAction line 164, col 26 to line 166, col 33 of module MongoRaftReconfig>
%   n1 :> "t1 S <<0, 0>> {n2, n3} (2,2)" @@
%   n2 :> "t3 S <<0, 0>> {n1} (2,0)" @@
%   n3 :> "t3 P <<1, 1>> {n2, n3} (1,3)" )
% /\ committed = {[term |-> 0, entry |-> <<1, 0>>]}
% /\ activeConfigs = {n3}
%

\begin{proof}[Proof of Theorem \ref{thm:leader-completeness}]
  In all initial states, the logs of all servers are empty, so the theorem holds. If we assume Theorem \ref{thm:leader-completeness} holds currently, it could only be falsified in the next state via actions that elect a primary or commit a log entry: \textit{BecomeLeader} and \textit{CommitEntry}. 
  
\begin{itemize}
    \item $BecomeLeader(i, Q)$ elects a primary server $i$ in a configuration $C_i$ in term $term'[i]$ with a quorum of voters $Q$. For such an election to occur we know that $C_i$ must be active, by Lemma \ref{lemma:deactivated-cannot-reconfig-or-elect}. To falsify the lemma, there must be 
    some committed entry $E_j=(ind_j,t_j)$, such that $t_j < term'[i]$.
    and $log[i]$ does not contain $E_j$. Since $C_i$ is active, we know, by Lemma \ref{lemma:active-config-quorums-intersect-committed}, that there must be some server $n \in Q$ such that $E_j$ is in $log[n]$. If $n$ voted for $i$ to become primary, we know that $LogGeq(i,n)$ (defined in Algorithm \ref{alg:mrr-pseudocode-full}) must have been satisfied in the current state. If we let 
    \begin{align*}
      (ind_i,t_i)&=(Len(log[i]), LogTerm(i))\\
      (ind_n,t_n)&=(Len(log[n]), LogTerm(n))
    \end{align*}
    % By Lemma \ref{lemma:primary-term-gte-log-term}, we know that the term of a primary is $\geq$ the largest term in its log, so we know that $term'[i] \geq t_i$. So, if $t_j \geq t_i$, we are sure that $log[i]$ contained all entries committed in lesser terms. So, it must be that $t_j > t_i$.
    % \mytodo{Clarify log comparison definitions.}
    % So, for $E_i=(ind_i,t_i)$, the last entry of $log[i]$, it must be that $(ind_i,t_i) \geq (ind,t)$, since the quorum required for $BecomeLeader(i)$ requires that the last entry of $i$ is newer all servers in some quorum $Q_i$ of $C_i$. So, $Q_i$ must have contained some server containing $E$ in its log, and we know that log entry terms increase monotonically (Lemma \ref{lemma:log-terms-monotonic}). Now, consider the following cases:
    % So, by Lemma \ref{lemma:newer-logs-contain-committed}, this implies that $log[i]$ contains all committed entries in terms $\leq t_i$. 
    there are two cases to consider:
    \begin{itemize}
        \item $t_i = t_n \wedge ind_i \geq ind_n$\\
        If $ind_i = ind_n$, then we have $log[i][ind_i]=t_i=t_n=log[n][ind_i]$. By Lemma \ref{lemma:log-matching}, this tells us that $log[i][..ind_i] = log[n][..ind_i]$, implying that $log[i]$ contains $E_j$, since $log[n]$ contains it. So, we consider the case where $ind_i > ind_n$. If $log[i][ind_i]=t_i$, and $log[n][ind_n]=t_n=t_i$, then, by Lemma \ref{lemma:uniform-log-entries}, this implies that $log[i][ind_n]=t_i = log[n][ind_n]$. So, by Lemma \ref{lemma:log-matching}, it must be that $log[i][..ind_n] = log[n][..ind_n]$. Since we know that $log[n]$ contains the committed entry $E_j$, this must mean that $log[i]$ contains it, contradicting our assumption that it did not contain $E_j$.
        
        \item $t_i > t_n$\\
        First, it must be that $t_n \geq t_j$, since the last entry of $log[n]$ has term $t_n$, we know that $log[n]$ contains entry $E_j$ in term $t_j$, and log entry terms increase monotonically (Lemma \ref{lemma:log-terms-monotonic}). So, we have the following:
        \begin{align*}
          t_j \leq t_n < t_i
        \end{align*}
        From Lemma \ref{lemma:newer-logs-contain-committed}, we know that $log[i]$ contains all entries committed in terms $< t_i$. So, it contains all entries committed in terms $t_j$, contradicting our assumption that it did not contain $E_j$ in its log.
    \end{itemize}

    \item $CommitEntry(i)$ commits a log entry $E$ in $term[i]$ on a primary server $i$. Assume there is some other server $j \neq i$, in configuration $C_j$, that is currently primary in term $t_j > term[i]$ and $E$ is not contained in $log[j]$. From Lemma \ref{lemma:primary-term-equals-cfg-term}, we know that $C_j.t=t_j$. And, by Lemma \ref{lemma:newer-configs-disable-commits}, we know that the existence of $C_j$ prevents commits of log entries occurring in any terms $< t_j$. So, such a $CommitEntry(i)$ entry action could not occur.
\end{itemize}
\end{proof}

\section{Detailed Model Checking Results}
\label{appendix:model-checking-details}

\label{sec:correctness-analysis}
In this section we provide additional details of our automated verification results using the TLC model checker. We first give a brief overview of TLC and its mode of operation, and then provide more detailed results from checking the safety properties discussed in Section \ref{sec:model-checking-section}.

\subsection{The TLC Model Checker}
TLC is an explicit state model checker that can check temporal properties of a given TLA+ specification. It is provided as a Java program that takes as input a TLA+ module file, a model checker configuration file, and a set of command line parameters. For checking safety properties, TLC assumes a TLA+ specification of the form $Init \wedge \square [Next]_{vars}$. The configuration file tells TLC the name of the specification to check and of the properties to be checked. In addition, the configuration file defines a \textit{model} of the specification, which is an assignment of values to any constant parameters of the specification.  It is also possible to provide a \textit{state constraint}, which is a state predicate that can be used to constrain the set of reachable states. If TLC discovers a reachable state that violates the state constraint predicate, it will not add the state to its current graph of reachable states. TLC also allows definition of a \textit{symmetry set}, which causes the model checker to consider states that have the same constant value under some permutation as equivalent, which can significantly reduce the set of reachable states for certain models \cite{Clarke1998}. A more complete and in-depth explanation of TLC behavior and parameters can be found in \cite{lamport2002specifying}. For all model checking runs discussed below we used TLC version 2.15 (adc67eb) running on CentOS Linux 7, with a 48-core, 2.30GHz Intel Xeon Gold 5118 CPU.

\subsection{Details and Results}
\label{sec:mrr-model-checking}

For checking safety of \textit{MongoRaftReconfig} we used a model we refer to as \textit{MCMongoRaftReconfig}, which imposes finite bounds on the \textit{MongoRaftReconfig} TLA+ specification. The complete, runnable TLC configuration for this model can be found in the supplementary materials \cite{supp-materials}. The model sets $Server=\{n1,n2,n3,n4\}$, and imposes the following state constraint:
\begin{align*}
\small
    StateConstraint \defeq \A s &\in Server : \\
                    &\wedge currentTerm[s] \leq MaxTerm \\
                    &\wedge Len(log[s]) \leq MaxLogLen \\
                    &\wedge configVersion[s] \leq MaxConfigVersion
\end{align*}
This constraint, along with a finite \textit{Server} set, is sufficient to make the reachable state space of this model finite, since it limits the size of the three unbounded variables of the specification: \textit{terms}, \textit{logs}, and configuration \textit{versions}. It restricts logs to be of a maximum finite length, and imposes a finite upper bound on terms and configuration versions. Figure \ref{mrr-model-checking-results} shows the parameters and results for this model. \textit{Permutation} is an operator in the \textit{TLC.tla} standard module \cite{tlc-tla-module} defined as the set of all permutations of elements in a given set. Under our symmetry declaration, any two states that are equal up to permutation of server identifiers are treated as equivalent by the model checker.

\begin{figure}
    \captionsetup[subfigure]{justification=centering}
    \centering
    \begin{subfigure}[t]{0.49\textwidth}
    \centering
    \small
    %
    % Results data: https://github.com/will62794/logless-reconfig/blob/4ee4a14b906191aa821869c5cd709819eb95c67c/tlc-results/2021-11-19_1637332509_MCMongoRaftReconfig.out
    %
    \begin{tabular}{|c|c|}
    \hline
    \multicolumn{2}{|c|}{\textit{MCMongoRaftReconfig}}\\
    % & \textit{MCMongoRaftReconfig}  \\
    \hline
     \textit{Server} & $\{n1,n2,n3,n4\}$ \\
    \hline
    \textit{MaxLogLen} & $2$  \\
     \hline
     \textit{MaxTerm} & $3$  \\
     \hline
     \textit{MaxConfigVersion} & $3$  \\
     \hline
     Constraint & \textit{StateConstraint}  \\
     \hline
     Symmetry & \textit{Permutation(Server)} \\
     \hline
     Invariant & \textit{LeaderCompleteness} \\
     \hline
     States & 345,587,274\\
     \hline
     Depth & 45\\
     \hline
     TLC Workers & 20\\
     \hline
     Duration & 8h 06min \\ 
     \hline
    \end{tabular}
    \caption{}
    \label{mrr-model-checking-results}
    \end{subfigure}
    %%%%%%%%%%%%%%%%%%%%%%%%%%%%%%%%%%%%
    \begin{subfigure}[t]{0.48\textwidth}
    \centering
    \small
    %
    % Results data: https://github.com/will62794/logless-reconfig/blob/338942abd0074f41acb9e9888edc5d9c064bc243/specs/tlc-results/2021-08-30_1630349623_MCMongoLoglessDynamicRaft.out
    %
    \begin{tabular}{|c|c|}
    \hline
    \multicolumn{2}{|c|}{\textit{MCMongoLoglessDynamicRaft}}\\
    \hline
     \textit{Server} & $\{n1,n2,n3,n4,n5\}$ \\
    \hline
     \textit{MaxTerm} & $4$  \\
     \hline
     \textit{MaxConfigVersion} & $4$  \\
     \hline
     Constraint & \textit{StateConstraint}  \\
     \hline
     Symmetry & \textit{Permutation(Server)} \\
     \hline
     Invariant & 
    \makecell{
          \textit{ElectionSafety} \\ 
     } \\
     \hline
     States & 812,587,401 \\
     \hline
     Depth & 30\\
     \hline
     TLC Workers & 20\\
     \hline
     Duration &19h 28min \\ 
     \hline
     %% To make vertical alignment consistent between the two tables.
    \multicolumn{2}{c}{}\\
    \end{tabular}
    \caption{}
    \label{mldr-model-checking-results}
    \end{subfigure}
    \caption{Summary of TLC Model Checking Results. \textit{States} is the number of reachable, distinct states discovered by TLC. \textit{Depth} is the length of the longest behavior.}
    \label{model-checking-results}
\end{figure}

As discussed in Section \ref{sec:model-checking-section}, the compositional structure of \textit{MongoRaftReconfig} makes it possible to verify \textit{MongoLoglessDynamicRaft} in isolation and assume that its safety properties hold in \textit{MongoRaftReconfig}. The full model checking results for our model, \textit{MCMongoLoglessDynamicRaft}, whose definition is provided in the supplementary materials \cite{supp-materials}, are presented in Figure \ref{mldr-model-checking-results}. 

\section{Subprotocol Refinement Proof}
\label{sec:subprotocol-refinement}

In order to ensure that any safety property of \textit{MongoLoglessDynamicRaft} holds for \textit{MongoRaftReconfig}, we must demonstrate that the behaviors of \textit{MongoLoglessDynamicRaft} are not augmented when operating as a subprotocol of \textit{MongoRaftReconfig}. To formalize this, we adopt TLA+ notation. Correctness properties and system specifications in TLA+ are both written as temporal logic formulas. This allows one to express notions of property satisfaction and refinement in a concise and similar manner. We say that a specification $S$ satisfies a property $P$ iff the formula $S \Rightarrow  P$ is valid (i.e. true under all assignments). We say that a specification $S_1$ refines (or is a refinement of ) $S_2$ iff $S_1 \Rightarrow  S_2$ is valid i.e. every behavior of $S_1$ is a valid behavior of $S_2$ \cite{Abadi1991}. So, if we view $MongoRaftReconfig$ and $MongoLoglessDynamicRaft$ as temporal logic formulas describing the set of possible system behaviors for each respective protocol, we can formally state our refinement theorem as follows:
\begin{theorem}
    $MongoRaftReconfig \Rightarrow MongoLoglessDynamicRaft$
\end{theorem}
To prove this, we must show that, (1) for any behavior $\sigma$ of \textit{MongoRaftReconfig}, the initial state of $\sigma$ is a valid initial state of \textit{MongoLoglessDynamicRaft} and (2) every transition in $\sigma$ is a valid transition of \textit{MongoLoglessDynamicRaft}. This proof has been formalized and checked in the TLA+ proof system. The full proof can be found in the supplementary materials \cite{supp-materials}. 

\fi

\end{document}